%% file: Paper_GLSM.tex
\documentclass[10pt]{iopartm}

\usepackage{calrsfs,amsmath,amssymb,graphicx,array,accents,citesort,natbib,figlatex,epsfig}
\usepackage{setspace,wrapfig,marvosym, bbm, stmaryrd,centernot, amsthm, color}

\newtheorem{theorem}{Theorem}[section]
\newtheorem{lemma}[theorem]{Lemma}

\newtheorem{defn}{Definition}
\newtheorem{assumption}{Assumption}
\newtheorem{rem}{Remark}

\newcommand{\R}{\mathbb{R}}
\newcommand{\TR}{{\mathcal T}_R}

\input{macros}
\begin{document}

\begin{center}
\title[ ]{\textcolor{black}{Generalized linear sampling method for elastic-wave sensing of heterogeneous fractures}}

\author{Fatemeh Pourahmadian$^1$, Bojan B. Guzina$^1$ and Houssem Haddar$^2$}
\address{$^1$ Department of Civil, Environmental and Geo-Engineering, University of Minnesota, Minneapolis, USA}
\address{$^2$ INRIA Saclay Ile de France and Ecole Polytechnique (CMAP) Route de Saclay, F-91128, Palaiseau, France}
\ead{guzin001@umn.edu}

\end{center} 

\begin{abstract}
A theoretical foundation is developed for active seismic reconstruction of fractures endowed with spatially-varying interfacial condition (e.g.~partially-closed fractures, hydraulic fractures). The proposed indicator functional carries a superior localization property with no significant sensitivity to the fracture's contact condition, measurement errors, and illumination frequency. This is accomplished through the paradigm of the $F_\sharp$-factorization technique and the recently developed Generalized Linear Sampling Method (GLSM) applied to elastodynamics. The direct scattering problem is formulated in the frequency domain where the fracture surface is illuminated by a set of incident plane waves, while monitoring the induced scattered field in the form of (elastic) far-field patterns. The analysis of the well-posedness of the forward problem leads to an admissibility condition on the fracture's (linearized) contact parameters. This in turn contributes toward establishing the applicability of the $F_\sharp$-factorization method, and consequently aids the formulation of a convex GLSM cost functional whose minimizer can be computed without iterations. Such minimizer is then used to construct a robust fracture indicator function, whose performance is illustrated through a set of numerical experiments. For completeness, the results of the GLSM reconstruction are compared to those obtained by the classical linear sampling method (LSM). 
\end{abstract}


\section{Introduction} \label{sec1}

Most recent advancements in the waveform tomography of discontinuity surfaces reside in the context of acoustic and electromagnetic inverse scattering. Spurred by the early study in~\cite{Kress1995}, such developments include: i) the Factorization Method (FM)~\cite{Bouk2013,Cham2014}; ii) the Linear Sampling Method (LSM)~\cite{Has2013, Fiora2003} and MUSIC algorithms~\cite{Park2009, Park2015(2)}; iii) the subspace migration technique~\cite{Park2015}, and iv) the method of Topological Sensitivity (TS)~\cite{Guz2004, Bonnet2011, Park2013}. In general, the LSM and FM techniques are applicable to a wide class of interfacial conditions and inherently carry a superior localization property -- potentially leading to high-fidelity geometric reconstruction. These methods, however, may suffer from the sensitivity to measurement uncertainties.  In contrast the TS approach, that is inherently robust to noisy data, fails to adequately recover the shape of a scatterer at long illuminating wavelengths. The subspace migration methods offer another alternative for a high-fidelity reconstruction, even from partial-aperture data, while requiring some a priori knowledge about the geometry of a discontinuity surface. Among the aforementioned methods, the LSM has been applied to the problem of elastic-wave imaging of fractures with homogeneous (traction-free) boundary condition~\cite{Bour2013}, while the TS approach was recently extended to cater for qualitative elastodynamic sensing of fractures endowed with a more general class of contact laws~\cite{Bellis2013, Fatemeh2015}. In geophysics, major strides~\cite{Willis2006,Zheng2013,Minato2013,Minato2014,Fang2014} have been made toward a robust reconstruction of fractures via seismic waveform tomography. So far the proposed methods, often reliant upon a rudimentary parameterization of the fracture geometry (e.g. planar fractures) and nonlinear minimization, entail a number of impediments including: i) high computational cost; ii) sensitivity to the assumed parametrization; iii) computational instabilities~\cite{Minato2014}, and iv) major restrictions in terms of the seismic sensing configuration~\cite{Fang2014,Minato2013}, namely the location of sources and receivers relative to the (planar) fracture surface. One recent study aiming to mitigate such limitations can be found in~\cite{Zheng2013} that makes use of focused Gaussian beams emitted from the surface source/receiver arrays to non-iteratively assess the orientation, spacing, and compliance of systems of parallel planar fractures.

This work aims to develop a non-iterative, full-waveform approach to 3D elastic-wave imaging of fractures with non-trivial (generally heterogeneous and dissipative) interfacial condition. To this end, the sought indicator map -- targeting \emph{geometric} fracture reconstruction -- is preferably (i) agnostic with respect to the fracture's interfacial condition, (ii) robust against measurement errors, and (iii) flexible in terms of sensing parameters, e.g. the illumination frequency. This is pursued by drawing from the theories of inverse scattering~\cite{Fiora2008, Col1992} and, in particular, by building upon the Factorization Method~\cite{Kirsch2008, Bouk2013} and the recently developed Generalized Linear Sampling Method (GLSM)~\cite{Audibert2015, Audibert2014} which completes the theoretical foundation of its LSM~predecessor. First, the inverse problem is formulated in the frequency domain where the illuminating wavefield is described by the elastic Herglotz wave function~\cite{Dassios1995} with its inherent compressional (P) and shear (S) wave components. On characterizing the induced scattered wavefield in terms of its far-field P- and S-wave patterns~\cite{Martin1993}, the far-field operator~$F$ is then defined as a map from the Herglotz densities to the far-field measurements. In this setting, the GLSM indicator functional is introduced as in~\cite{Audibert2014} on the basis of (i) a custom factorization of the far-field operator, and (b) a sequence of approximate solutions to the LSM integral equation, seeking Herglotz densities whose far-field pattern matches that of a point-load solution radiating from the sampling point. The latter sequence is essentially a set of penalized least-squares misfit functionals -- aimed at producing nearby solutions to the LSM equation, where the penalty term is constructed using a factorization component of~$F$. Minimizing this class of cost functionals in their most general form requires an optimization procedure~\cite{Audibert2014}. Thanks to the premise of a linear contact law, however, this study takes advantage of the so-called $F_\sharp$-factorization~\cite{Kirsch2008, Bouk2013} of the far-field operator to formulate the penalty term. This results in a sequence of \emph{convex} GLSM cost functionals whose minimizers can be computed without iterations. 


\section{Problem statement}\label{PS}
With reference to Fig.~\ref{fig1}(a), consider the elastic-wave sensing of a partially closed fracture $\Gamma \subset \mathbb{R}^3$ embedded in a homogeneous, isotropic, elastic solid endowed with mass density~$\rho$ and Lam\'{e} parameters $\mu$ and~$\lambda$. The fracture is characterized by a heterogeneous contact condition synthesizing the spatially-varying nature of its rough and/or multi-phase interface. Next, let $\Omega$ denote the unit sphere centered at the origin. For a given triplet of vectors $\bd\in\Omega$ and~$\bq_p,\bq_s\!\in\mathbb{R}^3$ such that $\bq_p\!\parallel\bd$ and~$\bq_s\!\perp\!\bd$, the obstacle is illuminated by a combination of compressional and shear plane waves 
\beq\lb{plwa}
\bu\ff(\bxi) ~=~ \bq_p \exs e^{\textrm{i} k_p \bxi \cdot \bd} \:+\: \bq_s \exs e^{\textrm{i} k_s \bxi \cdot \bd}
\eeq
propagating in direction~$\bd$, where $k_p$ and $k_s=k_p\sqrt{(\lambda\!+\!2\mu)/\mu}$ denote the respective wave numbers. The interaction of $\bu\ff$ with $\Gamma$ gives rise to the scattered field $\bv\in H^1_{\mathrm{loc}}(\R^3\backslash\Gamma)^3$, solving 
\beq\lb{GE}
\begin{aligned}
&\nabla \sip (\bC \colon \! \nabla \bv) \,+\, \rho \exs \omega^2\bv ~=~ \bzero \quad &\text{in}& \quad {\R^3}\backslash\Gamma, \\*[1mm]
&\bn \cdot \bC \exs \colon \!  \nabla  \bv~=~ \mathcal{L}(\dbv)  \,-\, \bt\ff  \quad &\text{on}& \quad \Gamma,
\end{aligned}      
\eeq
where $\omega^2=k_s^2 \mu/\rho$ is the frequency of excitation; $\dbv=[\bv^+\!-\bv^-]$ is the jump in~$\bv$ across~$\Gamma$, hereon referred to as the fracture opening displacement \textcolor{black}{(FOD)}; 
\beq\label{bC}
\bC \:=\: \lambda\,\bI_2\!\otimes\bI_2 \:+\: 2\mu\,\bI_4 
\eeq
is the fourth-order elasticity tensor; $\bI_m \,(m\!=\!2,4)$ denotes the $m$th-order symmetric identity tensor; \mbox{$\bt\ff = \bn \cdot \bC \colon \! \nabla \bu\ff$} is the free-field traction vector; $\bn = \bn^-$ is the unit normal on~$\Gamma$, and $\mathcal{L}: H^{1/2}(\Gamma)^3\to H^{-1/2}(\Gamma)^3$ represents a heterogeneous bijective contact law over the fracture surface, physically relating the displacement jump to surface traction. In many practical situations, the fracture's contact law is \emph{linearized} about a dynamic equilibrium state as 
\beq 
\label{contact} \mathcal{L}(\dbv) \:=\: \bK(\bxi) \dbv, \qquad \bxi \in \Gamma, 
\eeq
where $\bK=\bK(\bxi)$ is a \emph{symmetric} (due to reciprocity considerations) and possibly \emph{complex-valued} matrix of specific stiffness coefficients. 

\begin{figure}[tp]
\center\includegraphics[width=0.94\linewidth]{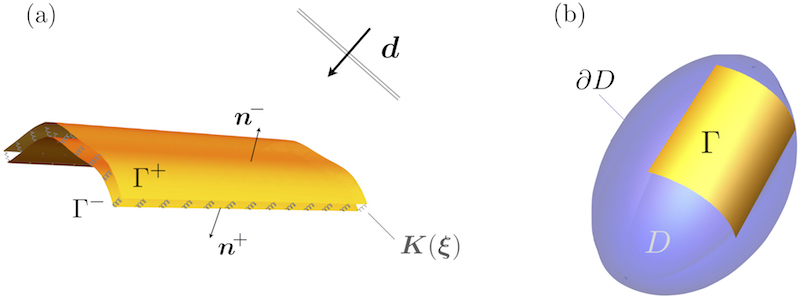} \vspace*{0mm} 
\caption{Direct scattering problem. The fracture boundary $\Gamma$ is arbitrarily extended to a piecewise smooth, simply connected, closed surface $\partial D$ of a bounded domain~$D$.} \lb{fig1}
\end{figure}

\begin{rem}
\textcolor{black}{In what follows, the analysis is based on the linear contact condition~(\ref{contact}) over~$\Gamma$. Under the premise of bijectivity, most of the ensuing developments (except for the $F_\sharp$ factorization method) can be adapted to handle nonlinear contact laws; such extension, however, is beyond the scope of this study.}
\end{rem}


The formulation of the direct scattering problem can now be completed by requiring that~$\bv$ satisfies the Kupradze radiation condition at infinity \cite{Kuprad1979}. On uniquely decomposing the scattered field into an irrotational part and a solenoidal part as $\bv = \bv^p + \bv^{s}$ \textcolor{black}{where 
\beq\label{vpvs}
\bv^p = \frac{1}{k_s^2\!-\!k_p^2}(\Delta+k_s^2)\bv, \qquad \bv^{s} = \frac{1}{k_p^2\!-\!k_s^2}(\Delta+k_p^2)\bv, 
\eeq}
the Kupradze condition can be stated as 

\beq\lb{KS}
\frac{\partial\bv^p}{\partial r} - \text{i} k_p \bv^p = o\big(r^{-1}\big) \quad \mbox{ and } \quad 
\frac{\partial\bv^s}{\partial r} - \text{i} k_s \bv^ s = o\big(r^{-1}\big) \qquad \text{as} ~~r:=|\bxi|\to\infty,
\eeq    
uniformly with respect to $\hat\bxi:=\bxi/r$. 

 \paragraph*{Dimensional platform.}
In what follows, all quantities are rendered \emph{dimensionless} by taking $\rho$, $\mu$, and ${\sf R}$ -- the characteristic size of a region sampled for fractures -- as the respective scales for mass density, elastic modulus, and length -- which amounts to setting $\rho = \mu = {\sf R} = 1$~\cite{Scaling2003}.

\paragraph*{Function spaces.}
To assist the ensuing analysis, the fracture surface $\Gamma$ is arbitrarily extended, as shown in Fig.~\ref{fig1}(b), to a piecewise smooth, simply connected, closed surface $\partial D$ of a bounded domain $D$ such that the normal vector $\bn$ to the fracture surface $\Gamma$ coincides with the outward normal vector to $\partial D$ -- likewise denoted by $\bn$. We also assume that $\Gamma$ is an open set (relative to $\partial D$) with positive surface measure.  
Following \cite{McLean2000}, we define
\beq\lb{funS}
\begin{aligned}
&H^{\pm \frac{1}{2}}(\Gamma) ~:=~\big\lbrace f\big|_\Gamma \colon \,\,\, f \in H^{\pm \frac{1}{2}}(\partial D) \big\rbrace, \\*[1 mm]
& \tilde{H}^{\pm \frac{1}{2}}(\Gamma) ~:=~\big\lbrace  f \in H^{\pm\frac{1}{2}}(\partial D) \colon  \,\,\, \text{supp}(f) \subset \overline{\Gamma} \big\rbrace,
\end{aligned}
\eeq
and recall that $H^{-1/2}(\Gamma)$ and $\tilde{H}^{-1/2}(\Gamma)$ are respectively the dual spaces of $\tilde{H}^{1/2}(\Gamma)$ and $H^{1/2}(\Gamma)$. Accordingly, the following embeddings hold
\beq\lb{embb}
\tilde{H}^{\frac{1}{2}}(\Gamma) \,\subset\, H^{\frac{1}{2}}(\Gamma) \,\subset\, L^2(\Gamma) \,\subset\, \tilde{H}^{-\frac{1}{2}}(\Gamma) \,\subset\, H^{-\frac{1}{2}}(\Gamma).
\eeq

\begin{rem} \textcolor{black}{In the context of fracture mechanics, it is well known that $\dbv(\bxi)\to \boldsymbol{0}$ continuously as $\Gamma\!\ni\!\bxi\to\partial\Gamma$ (typically as $d^\alpha$, $0\!<\!\alpha\!\leqslant\!\tfrac{1}{2}$~\cite{Ueda2006} where $d$ is a normal distance to $\partial\Gamma$ when $\partial\Gamma$ is smooth), which lends credence to the assumption $\dbv\in\tilde{H}^{1/2}(\Gamma)^3$ used hereon.}           
\end{rem}

\section{On the well-posedness of the forward scattering problem} \label{WP}


Serving as a prerequisite for the analysis of the inverse scattering problem, this section investigates the well-posedness of the direct scattering problem \eqref{GE}--\eqref{KS}.  Let $R>0$ be sufficiently large so that the ball $B_R$ of radius $R$ contains $\Gamma$, and consider the Dirichlet-to-Neumann operator $\TR: H^{1/2}(\partial B_R)^3 \to H^{-1/2}(\partial B_R)^3$ associated with the scattering problem in $\R^3 \backslash B_R$, namely 
$$
\TR(\bphi)(\bxi) := \hat{\bxi}\cdot \bC \colon \! \nabla \btu_{\bphi}(\bxi), \qquad \bxi\in\partial B_R,
$$
where $\btu_\bphi \in H^1_\mathrm{loc}(\R^3 \backslash B_R)^3$ is the unique radiating solution, satisfying~\eqref{KS}, of 
\beq\lb{TR}
\begin{aligned}
&\nabla \sip (\bC \colon \! \nabla \btu_\bphi) \,+\, \rho \exs \omega^2 \btu_\bphi ~=~ \bzero     \quad &\text{in}& \quad {\R^3}\backslash B_R, \\*[1mm]
&  \btu_\bphi~=~\bphi \quad  &\text{on}& \quad \partial B_R. 
\end{aligned}      
\eeq
The scattering problem \eqref{GE}--\eqref{KS} can now be equivalently written in terms of  $\bv\in H^1(B_R\backslash\Gamma)^3$ as 
\beq\lb{GEbis}
\begin{aligned}
&\nabla \sip (\bC \colon \! \nabla \bv) \,+\, \rho \exs \omega^2\bv ~=~ \bzero  \quad &\text{in}& \quad {\R^3}\backslash\Gamma, \\*[1mm]
&\bn \cdot \bC \exs \colon \!  \nabla  \bv ~=~  \bK\sip\dbv  \,-\, \bt\ff   \quad &\text{on}& \quad \Gamma, \\*[1mm]
&\bn \cdot \bC \colon \! \nabla \bv = \TR(\bv)  \quad &\text{on}& \quad \partial B_R, 
\end{aligned}      
\eeq
where $\bn(\bxi)=\hat\bxi$ on~$\partial B_R$. This problem can be written variationally in terms of $\bv\in H^1(\Bo\backslash\Gamma)^3$ as  
\beq\lb{Wik-GE}
\begin{aligned}
&- \rho \exs \omega^2 \! \int_{\Bo \backslash \Gamma} \overline\bw \cdot \bv  \,\exs \textrm{d}V_{\bxi}  ~+\, \int_{\Bo \backslash \Gamma}   \nabla \exs \overline{\bw} \colon \bC \colon \nabla \bv \, \textrm{d}V_{\bxi} ~+\,  \dualGA{\bK \sip \dbv}{\llbracket\bw\rrbracket}~-\,  \\*[1 mm]
& \quad\,\,  \dualBR{\TR(\bv)}{\bw}   ~=\,   \int_{\Gamma} \exs  \overline{\llbracket \bw \rrbracket} \cdot \bt\ff \, \textrm{d}S_{\bxi},   \qquad \quad \forall \exs \bw \in H^1(\Bo \backslash \Gamma)^3,
\end{aligned}
\eeq
where $\dualGA{\cdot}{\cdot}$ and $\dualBR{\cdot}{\cdot}$ respectively denote the $\big\langle H^{-1/2}(\Gamma)^3, \tilde{H}^{1/2}(\Gamma)^3 \big\rangle$ and $\big\langle H^{-1/2}(\partial B_R)^3, {H}^{1/2}(\partial B_R)^3 \big\rangle$ \emph{duality products} that extend $L^2$ inner products. The analysis of the forward scattering problem is based on the following properties of the Dirichlet-to-Neumann operator $\TR$ (see also \cite{BramblePasciak}). \textcolor{black}{For clarity, we will use an abbreviated notation of relevant vector norms where e.g. $\|\boldsymbol{\cdot}\|_{H^{1/2}(\Gamma)^3}$ is denoted by~ $\|\boldsymbol{\cdot}\|_{H^{1/2}(\Gamma)}$ and so on.}

\begin{lemma} \label{LemmaTR}
There exists a bounded, non-negative and self-adjoint operator $\TR^0\!: H^{1/2}(\partial B_R)^3 \to H^{-1/2}(\partial B_R)^3$ such that $\TR+\TR^0\!: H^{1/2}(\partial B_R)^3 \to H^{-1/2}(\partial B_R)^3 $ is compact.  Moreover, 
\begin{equation}
\Im \dualBR{\TR(\bphi)}{\bphi} > 0 \qquad \forall \bphi \in H^{1/2}(\partial B_R)^3:~ \bphi \neq 0.
\end{equation}
\end{lemma}

\begin{proof} Let $R_{\circ} > R$ and $\bphi, \bpsi \in H^{1/2}(\partial B_R)^3 $. Multiplying the first equation in~\eqref{TR} by $\overline{\btu_\bpsi}$ and
integrating by parts on $B_{R_{\circ}} \backslash B_R$ yields
$$
\dualBR{\TR(\bphi)}{\bpsi} ~=~ \rho\exs\omega^2 \! \int_{B_{R_{\circ}}\!\backslash B_R} \overline{\btu_\bpsi}\cdot \btu_\bphi  \,\exs \textrm{d}V_{\bxi}  ~- \int_{B_{R_{\circ}}\!\backslash B_R} \nabla \exs \overline{\btu_\bpsi} \colon \bC \colon \nabla \btu_\bphi
\, \textrm{d}V_{\bxi} ~+ \int_{\partial{B_{R_{\circ}}}}  \exs\overline{\btu_\bpsi} \cdot \bt(\bphi)  \,  \textrm{d}S_{\bxi}, 
$$
where $\bt(\bphi)(\bxi):= \hat{\bxi}\cdot \bC \colon \! \nabla \btu_{\bphi}(\bxi)$ for $\bxi \in \partial B_{R_{\circ}}$. Using the well-posedness of \eqref{TR} and the Riesz representation theorem, we define $\TR^0$ by
$$
\dualBR{\TR^0(\bphi)}{\bpsi} := \int_{B_{R_{\circ}}\!\backslash B_R} \nabla\exs\overline{\btu_\bpsi} \colon \bC \colon \nabla \btu_\bphi
\, \textrm{d}V_{\bxi}.
$$
\textcolor{black}{On demonstrating that $\|(\TR+\TR^0)(\bphi)\|_{H^{-1/2}(\Gamma)}\leqslant C (\|\btu_{\bphi}\|_{L^2(B_{R_{\circ}}\!\backslash B_R)}+\|\bt(\bphi)\|_{L^2(\partial B_{R_{\circ}}\!)})$ for some constant~$C\!>\!0$ independent of~$\bphi$, the compactness of $\TR+\TR^0$ then follows from the compactness of mapping $\bphi\to\btu_\bphi$ (resp. $\bphi\to\bt(\bphi)$) from $H^{1/2}(\partial B_R)$ into~$L^2(B_{R_{\circ}}\!\backslash B_R)$ (resp.~$L^2(\partial B_{R_{\circ}})$) thanks to the compact embedding of $H^1(B_{R_{\circ}}\!\backslash B_R)$ into~$L^2(B_{R_{\circ}}\!\backslash B_R)$ and the standard regularity results for scattering problems~\cite{McLean2000}, which can be recovered from the boundary integral representation of $\btu_\bphi$ in $\R^3\backslash B_R$ in terms of boundary data on~$\partial B_R$. As shown in~\ref{lem1-p}, the sign of the imaginary part of $\TR$ is a consequence of the asymptotic behavior of $\btu_\bphi$ at infinity~\cite{Kuprad1979} which implies}
\beq\label{kup1}
\Im \dualBR{\TR(\bphi)}{\bphi}  = \Im \lim_{R_{\circ} \to \infty} \int_{\partial{B_{R_{\circ}}}}  \exs\overline{\btu_\bphi} \cdot \bt(\bphi)  \,
\textrm{d}S_{\bxi} \;=\; \lim_{R_{\circ} \to \infty} \int_{\partial{B_{R_{\circ}}}}   \Big\lbrace  k_p (\lambda+2\mu) |\btu_\bphi^{\mathrm{p}} \exs |^2 \,+\, k_s\mu\exs |\btu_\bphi^{\mathrm{s}}|^2 \Big\rbrace \,\,  \textrm{d}S_{\bxi}.
\eeq
The \textcolor{black}{sign-definiteness} of the imaginary part is a consequence of the Rellich lemma \cite{Col1992} applied to $\btu_\bphi^{\mathrm{p}}$ and $\btu_\bphi^{\mathrm{s}}$, which requires that $\btu_\bphi = \btu_\bphi^{\mathrm{p}} + \btu_\bphi^{\mathrm{s}}=0$ whenever $\Im \dualBR{\TR(\bphi)}{\bphi}=0$.  
\end{proof}
\begin{theorem} \label{maindirect}
Assume that $\bt\ff \in H^{-1/2}(\Gamma)^3$ and that $\bK \in L^\infty(\Gamma)^{3\times 3}$ is symmetric such that $\Im\bK\leqslant\bzero$ on $\Gamma$, i.e. that $\overline{\boldsymbol{\theta}} \sip \Im \bK(\bxi)\sip \boldsymbol{\theta} \leqslant 0$, $\forall\boldsymbol{\theta}\in\mathbb{C}^3$ and a.e.~on~$\Gamma$. Then problem~\eqref{Wik-GE} has a unique solution that continuously depends on $\bt\ff \in H^{-1/2}(\Gamma)^3$.
\end{theorem}
\begin{proof}
Since $\bt\ff \in H^{-1/2}(\Gamma)$, the antilinear form $\int_{\Gamma} \exs \overline{\llbracket \bw \rrbracket} \cdot \bt\ff \, \textrm{d}S_{\bxi}$ may be understood as a duality pairing $\dualGA{\cdot}{\cdot}$. The continuity of this form comes from the continuity of the trace mapping $\bw \to \llbracket \bw \rrbracket$ from $ H^1(\Bo \backslash \Gamma)^3$ into $\tilde{H}^{1/2}(\Gamma)^3$.

On the basis of the adopted dimensional platform i.e.~$\rho = \mu = 1$ (see Section~\ref{PS}), the sesquilinear form on the left hand side of \eqref{Wik-GE} can be decomposed into a coercive part
\beq\lb{cv}
\begin{aligned} 
&\textrm{A}(\bv,\bw) ~=\,  \! \int_{\Bo \backslash \Gamma} \overline\bw \cdot \bv  \,\exs \textrm{d}V_{\bxi}  ~+\, \int_{\Bo \backslash \Gamma} \!\!  \nabla \exs \overline{\bw} \colon \! \bC \colon \! \nabla \bv \, \textrm{d}V_{\bxi}  ~+\,  \dualBR{\TR^0(\bv)}{\bw}  , \quad \forall \bw \in H^1(\Bo \backslash \Gamma)^3,
\end{aligned}
\eeq
and a compact part
\beq\lb{compact}
\text{B}(\bv,\bw) ~=\, - (1+k_s^2) \int_{\Bo \backslash \Gamma}  \overline\bw
\cdot \bv  \,\, \textrm{d}V_{\bxi} ~+\,  \dualGA{\bK \sip \dbv}{\llbracket\bw\rrbracket}~ -\,
\dualBR{(\TR+\TR^0)(\bv)}{\bw}   , \quad \forall \bw \in H^1(\Bo \backslash \Gamma)^3.
\eeq
The coercivity of $\textrm{A}(\bv,\bw)$ follows from the Korn inequality~\cite{McLean2000} and the non negative sign of $\TR^0$ (Lemma \ref{LemmaTR}). Now, in order to prove that the antilinear form $\text{B}$ defines a compact perturbation of $\textrm{A}(\bv,\bw)$, one may observe that 
\[
| \text{B}(\bv,\bw) | ~\leqslant~\! \text{c}_2 \exs \big\lbrace \nxs
\norms{\!\bv\!}_{L^2(\Bo \backslash \Gamma)} \, \norms{\!\bw\!}_{L^2(\Bo
  \backslash \Gamma)} \,+\, \norms{\! \llbracket  \bv  \rrbracket
  \!}_{L^2(\Gamma)} \, \norms{\! \llbracket  \bw  \rrbracket \!}_{L^2(\Gamma)}
\! \big\rbrace + \norms{(\TR+\TR^0)(\bv)}_{H^{-1/2}(\partial B_R)}\norms{\bw}_{H^{1/2}(\partial B_R)} 
\]
for a constant $\text{c}_2$ independent of $\bv$ and $\bw$. The claim then follows from Lemma \ref{LemmaTR},  the compact embedding of $H^1(\Bo \backslash \Gamma)$ into $L^2(\Bo \backslash \Gamma)$ and the compactness of the trace operator $\bv \rightarrow  \llbracket \bv \rrbracket$ as an application from $H^1(\Bo \backslash \Gamma)$ into ${L^2(\Gamma)}$ where the latter comes from the compact embedding of $\tilde H^{1/2}(\Gamma)$ into $L^2(\Gamma)$.

Problem \eqref{Wik-GE} is then of Fredholm type, and is therefore well-posed as soon as the uniqueness of a solution is guaranteed. Assume that $\bt\ff = 0$. Then 
$$
\Im \dualBR{\TR(\bv)}{\bv} =\, \dualGA{\Im \bK \sip \dbv}{\llbracket\bv\rrbracket} \,\leqslant\: 0 
$$
by premise of the Theorem. Thanks to Lemma  \ref{LemmaTR}, this requires that $\bv\!=\!\boldsymbol{0}$ on $\partial B_R$ and thus $\bv\!=\!\boldsymbol{0}$ in $\Bo \backslash \Gamma$ by the unique continuation principle.
 \end{proof}

\section{Elements of the inverse scattering solution}\lb{Prelim}   
\renewcommand{\OOd}{{\Omega}}
\renewcommand{\pff}{\btu}

This section is devoted to the introduction of the \emph{far-field operator} -- relevant to the scattering problem~(\ref{GE}), and the derivation of its first and second factorizations. In the sequel, we assume that the hypotheses of Theorem \ref{maindirect} hold.

\paragraph*{Elastic Herglotz wave function.} For given density $\bg \in L^2(\OOd)^3$, we consider the unique decomposition 
\beq\lb{herden}
\bg \;:=\; \bg_p \oplus\, \bg_s 
\eeq
such that $\bg_p(\bd)\!\parallel\!\bd$ and $\,\bg_s(\bd)\!\perp\!\bd$, $\,\bd\in\OOd$. In dyadic notation, one has 
\beq\lb{freef}
\bg_p(\bd) := (\bd \nxs \otimes \nxs \bd \exs) \cdot \bg(\bd) \quad ~~\text{and}~~ \quad   \bg_s(\bd) := (\bI - \bd \nxs \otimes \nxs \bd \exs) \cdot \bg(\bd).
\eeq 
Next, we define the elastic Herglotz wave function~\cite{Dassios1995} as 
\beq\lb{HW}
\pff_\bg(\bxi) ~: =~  \int_{\OOd} \bg_p(\bd) \exs  e^{\textrm{i} k_p \bd \cdot \bxi} \,\, \text{d}S_{\bd} \,+\, \int_{\OOd} \bg_s(\bd) \exs  e^{\textrm{i} k_s \bd \cdot \bxi} \,\, \text{d}S_{\bd},  \qquad \bxi \in \R^3 
\eeq
in terms of {the compressional and shear wave densities~$\bg_p$ and~$\bg_s$.
\paragraph*{The far-field pattern.}
As shown in~\cite{Martin1993}, any scattered wave $\bv\in H^1_{\mathrm{loc}}(\R^3 \backslash \Gamma)^3$ solving \eqref{GE}-\eqref{KS} has the asymptotic
expansion
\beq\lb{vinf}
\bv(\bxi) ~=~ \exs \frac{e^{\text{i}k_p r}\!}{4 \pi(\lambda\!+\!2\mu)r} \exs \bv^\infty_p(\hat\bxi) \:+\: 
\frac{e^{\text{i}k_s r}\!}{4\pi\mu r} \exs \bv^\infty_s(\hat\bxi) \:+\: O(r^{-2}) \quad~~ \text{as} \quad r:=|\bxi|\to\infty, 
\eeq
where $\hat\bxi$ is the unit direction of observation, while $\bv^\infty_p$ and $\bv^\infty_s$ denote respectively the far-field patterns of $\bv^p$ and $\bv^s$ -- see~\eqref{vpvs}, which satisfy  $\bv^\infty_p\!\parallel\hat\bxi$ and $\bv^\infty_s\!\perp\hat\bxi$.  In this setting, we define the far-field pattern of~$\bv$ by
\beq\lb{far-field}
\bv^\infty := \bv^\infty_p \oplus \bv^\infty_s.
\eeq
By way of the integral representation theorem in elastodynamics~\cite{Bon1999} and the far-field representation of the elastodynamic fundamental stress tensor (see Appendix), one can show that if $\bv\in H^1_{\mathrm{loc}}(\R^3 \backslash \Gamma)^3$ satisfies~\eqref{GE}-\eqref{KS}, then 
\beq\lb{vinf2}
\begin{aligned}
&\bv_p^\infty(\hat\bxi) ~=~- \text{i} k_p \exs \hat\bxi \int_\Gamma \Big\lbrace \lambda \,  \llbracket \bv \rrbracket \sip \bn   + 2\mu \big(\bn \sip \hat\bxi \big) \exs \llbracket \bv \rrbracket \sip \hat\bxi  \exs \Big\rbrace \, e^{-\text{i}k_p \hat\bxi \cdot \bx}  \,\, \text{d}S_{\bx}, \\*[1 mm]
& \bv_s^\infty(\hat\bxi) ~=~ -\text{i}k_s \exs \hat\bxi \exs \times \int_{\Gamma} \Big\lbrace \mu \big( \llbracket \bv \rrbracket \!\times\! \hat\bxi \exs \big)(\bn \sip \hat\bxi \exs ) \,+\, \mu \big( \bn \!\times \hat\bxi \big) (\llbracket \bv \rrbracket \sip \hat\bxi) \Big\rbrace \, e^{-\text{i}k_s \hat\bxi \cdot \bx} \,\, \text{d}S_{\bx}.
\end{aligned} 
\eeq
\paragraph*{The far-field operator.} 

\begin{defn}\lb{deffarf}
We define the far-field operator $F: L^2(\OOd)^3 \to L^2(\OOd)^3$ by
\beq\lb{ffo0}
F(\bg) ~=~ \textcolor{black}{\bv_{\bg_\Omega}^\infty,} 
\eeq
where~$\bv_{\bg_\Omega}^\infty$ is the far-field pattern~\eqref{far-field} of~$\bv\in H^1_{\mathrm{loc}}(\R^3 \backslash \Gamma)^3$ solving~\eqref{GE}-\eqref{KS} with data $\bu\ff = \pff_\bg$, see~\eqref{HW}.  
\end{defn}

When the contact law specified by~$\mathcal{L}(\dbv)$ is linear as in~(\ref{contact}), the far-field operator can be expressed as a linear integral operator. To examine this case, \textcolor{black}{consider an incident plane wave~(\ref{plwa}) propagating in direction $\bd\in\OOd$ with amplitude $\bq=\bq_p\oplus\bq_s$}, and denote the induced far-field pattern~(\ref{far-field}) by \textcolor{black}{$\bv^\infty_\bq(\bd, \cdot)=\bv^\infty_{\bq_p}\oplus \bv^\infty_{\bq_s}$}. Next, let us define the far-field kernel  \textcolor{black}{$\bW^\infty(\bd,\hat\bxi)\in\mathbb{C}^{6\times 6}$} so that 
\beq\lb{w-inf}
\bW^\infty(\bd,\hat\bxi) \sip \bq ~:=~ \bv^\infty_\bq(\bd, \hat\bxi). 
\eeq
Then one easily verifies that  
\beq\lb{ffo2} 
 F(\bg)\textcolor{black}{(\hat\bxi)} ~=\,  \int_{\OOd} \bW^\infty(\bd,\hat\bxi) \sip \bg(\bd) \,\, \text{d}S_{\bd}.
\eeq 

\begin{lemma}\lb{recip}
The far-field kernel $\bW^\infty(\bd,\hat\bxi)$ satisfies the reciprocity identity
\beq\lb{W-recip}
\textcolor{black}{
\bW^\infty(\bd,\hat\bxi) ~=~ \overline{\bW^{\infty *}}(-\hat\bxi,-\bd), \qquad \forall\bd,\hat\bxi\!\in\OO.}
\eeq 
\end{lemma}
\begin{proof}
\textcolor{black}{See~\ref{Recip}}. 
\end{proof}

\section{Key properties for the application of sampling methods}\label{SFS}

\paragraph*{Factorization of the far-field operator $F$.}
Consider the Herglotz operator $\mathcal{H} \colon L^2(\OOd)^3 \rightarrow H^{-1/2}(\Gamma)^3$ given by 
\beq\lb{oH}
\mathcal{H}(\bg) ~:=~ \bn\cdot\bC\exs\colon\!\nabla\pff_\bg \quad~~ \text{on}\quad \Gamma,
\eeq   
where $\pff_\bg$ is the Herglotz wave function~(\ref{HW}). Next, define $\mathcal{G} \colon H^{-1/2}(\Gamma)^3 \rightarrow L^2(\OOd)^3$ as the map taking the traction vector $\bt\ff$ over $\Gamma$ to the far-field pattern, $\bv^\infty$, of $\bv\in H^1_{\mathrm{loc}}(\R^3 \backslash \Gamma)^3$ satisfying \eqref{GE}-\eqref{KS}. Then from Definition~\ref{deffarf}, the far-field operator~\eqref{ffo0} becomes 
\beq\lb{fac1}
F ~=~ \mathcal{G} \mathcal{H}.
\eeq   
\begin{lemma}\lb{H*}
With reference to decomposition~\eqref{far-field}, the adjoint Herglotz operator $\mathcal{H}^* \colon \tilde{H}^{1/2}(\Gamma)^3 \rightarrow L^2(\OOd)^3$ takes the form
\beq\lb{Hstar}
\begin{aligned}
\mathcal{H}^*\nxs(\ba)(\hat\bxi) ~=~ -  \Big(  \, & \text{\emph{i}} k_p \,  \hat\bxi \exs \int_\Gamma \, \big\lbrace \lambda \exs (\ba \sip \bn) \,+\, 2\mu \exs (\bn \sip \hat\bxi) ( \ba \sip \hat\bxi)  \big\rbrace  \, e^{-\text{\emph{i}}k_p \hat\bxi \cdot \by} \,\, \text{d}S_{\by}  \\*[1 mm]
 &  \textcolor{black}{\oplus}~ \text{\emph{i}} k_s \,  \hat\bxi\times\! \int_\Gamma \big\lbrace   \mu \exs(\ba \times \hat\bxi)(\bn\sip\hat\bxi) \,+\, \mu \exs  (\bn \times \hat\bxi) (\ba \sip \hat\bxi)  \big\rbrace \, e^{-\textrm{\emph{i}} k_s \hat\bxi \cdot \by} \,\, \text{d}S_{\by} \Big).
\end{aligned}
\eeq
\end{lemma}
\begin{proof}
\textcolor{black}{see~\ref{H*pruf}}.
\end{proof}
On the basis of~\eqref{vinf2} and~(\ref{Hstar}), map $\mathcal{G}$ can be further decomposed as $\mathcal{G}=\mathcal{H}^* T$ where the middle operator $T\colon H^{-1/2}(\Gamma)^3 \rightarrow \tilde{H}^{1/2}(\Gamma)^3$ is given by
\beq\lb{T}
T(\bt\ff)(\bxi) ~:=~ \llbracket \bv(\bxi) \rrbracket,   \qquad \bxi \in \Gamma
\eeq
such that~$\bv\in H^1_{\mathrm{loc}}(\R^3 \backslash \Gamma)^3$ satisfies~\eqref{GE}-\eqref{KS} or equivalently \eqref{Wik-GE}. Thanks to this new decomposition of $\mathcal{G}$, the second factorization of $F \colon L^2(\OOd)^3 \rightarrow L^2(\OOd)^3$ is obtained 
\beq\lb{fact} 
F ~=~ \mathcal{H}^* \exs T \exs \mathcal{H}, 
\eeq
which provides the second important ingredient for the ensuing analysis.  

\paragraph*{Properties of the Herglotz operator $\mathcal{H}$.}

\begin{lemma}\lb{H*p}
Operator $\mathcal{H^*}:\tilde{H}^{1/2}(\Gamma)^3\rightarrow L^2(\OOd)^3$ in Lemma~\ref{H*} is compact and injective.  
\end{lemma}
\begin{proof} 
Integral operator $\mathcal{H}^*$ has a smooth kernel and is therefore compact from $\tilde{H}^{1/2}(\Gamma)^3$ into $L^2(\OOd)^3$. Next, suppose that there exists $\ba \in \tilde{H}^{1/2}(\Gamma)^3$ such that $\mathcal{H^*}(\ba) = \bzero$. In light of~(\ref{vinf}) and~(\ref{vinf2}), it is apparent that $\mathcal{H^*}$ is nothing else but the far-field operator stemming from the double-layer potential
\beq\lb{Dlp}
\bV(\ba)(\bxi) ~=~ \int_{\Gamma} \ba(\by) \cdot  \bfT(\bxi,\by)  \, \text{d}S_{\by}, \qquad 
\bfT(\bxi,\by) ~=~ \bn(\by)\cdot\bSig(\bxi,\by), \qquad \bxi \in \Bo \backslash \Gamma,
\eeq
where 
$\bSig(\bxi,\by)$ is the (third-order) elastodynamic fundamental stress tensor \textcolor{black}{given in~\ref{stress-fund}}. By virtue of definition~(\ref{vinf}), vanishing far-field pattern of $\bV(\ba)$ implies, by the Rellich Lemma and the unique continuation principle, that $\bV(\ba) = \bzero$ in $\R^3 \backslash \Gamma$. Owing to the fundamental jump property of double-layer potentials by which~$\llbracket \bV \rrbracket=\ba$, one obtains $\ba = \bzero$ which guarantees the injectivity of $\mathcal{H}^*$.  
\end{proof}

One additional property that is needed for the analysis of sampling methods is the densness of the range of~$\herg^*$, which is equivalent to the injectivity of $\herg$. Unfortunately the latter cannot be guaranteed in general, and one has to impose this property as an assumption on $\Gamma$ and $\omega$. 

\begin{assumption}\label{Inject-H}
We assume that $\Gamma$ and $\omega$ are such that the Herglotz operator $\herg\!:L^2(\OOd)^3\rightarrow H^{-1/2}(\Gamma)^3$ is injective, i.e. that $\,\mathcal{H}^*\!:\tilde{H}^{1/2}(\Gamma)^3\rightarrow L^2(\OOd)^3$ has a dense range. 
\end{assumption}

The following lemma indicates why we expect that for a given fracture geometry $\Gamma$, Assumption \ref{Inject-H} holds for all $\omega\!>\!0$ \textcolor{black}{possibly} excluding a discrete set of values without finite accumulation points. 

\begin{lemma}\lb{Dense-H}
Assume that \textcolor{black}{$\exs\Gamma$ contains $M\!\geqslant\!1$ (possibly disjoint) analytic surfaces~$\Gamma_m\!\subset\Gamma$, $m=1,\ldots M$, and consider the unique analytic continuation $\partial D_m$ of $\:\Gamma_m$ identifying ``interior'' domain~$D_m\!\subset\mathbb{R}^3$}. Then Assumption~\ref{Inject-H} holds as soon as \textcolor{black}{for any such~$m$, $\omega\!>\!0$} is not a ``Neumann'' eigenfrequency of the Navier equation in $D_m$, i.e. as long as every function $\pff \in H^1(D_m)^3$ satisfying
\beq\lb{uiH}
\begin{aligned}
&\nabla \sip (\bC \colon \! \nabla \pff) \,+\, \rho \exs \omega^2 \pff ~=~ \bzero     \quad  &\textrm{in}~ D_m, \\*[1 mm]
&\bn \sip \bC \colon \! \nxs \nabla  \pff ~=~ \bzero   \quad &\textrm{on}~ \partial D_m
\end{aligned}
\eeq
\textcolor{black}{vanishes identically in $D_m$. Further if~$D_m$ is bounded, the real eigenfrequencies of~\eqref{uiH} form a discrete set.} 
\end{lemma}
\begin{proof}
\textcolor{black}{Let~$\Gamma_m$ denote the $m$th analytic piece of~$\Gamma$. Recalling~\eqref{HW} and invoking the analyticity of $\bn\cdot\bC\colon \!\nabla \pff_\bg$ with respect to the surface coordinates on $\partial D_m$, we deduce that if $\bn \sip \bC\colon \!\nabla \pff_\bg=\bzero$ on $\Gamma_m\subset\partial D_m$} then
$$
\bn \cdot \bC \exs \colon  \!\nabla\pff_\bg =\bzero \quad \mbox{ on}~ \partial D_m.
$$
This means that $\pff_\bg =\bzero$ in $D_m$ since $\omega$ is not a ``Neumann'' eigenvalue of the Navier equation in $D_m$. The unique continuation principle then implies that $\pff_\bg\!=\!\bzero$ in $\R^3$. Accordingly, we deduce that the Herglotz density vanishes, i.e. that~$\bg\!=\!\bzero$ as in the scalar case \cite{Col1992}. \textcolor{black}{The proof of discreteness of the set of real eigenfrequencies characterizing~\eqref{uiH} when~$D_m$ is bounded can be found in~\cite{Kuprad1979}, Chapter~7, Theorem~1.4}. 
\end{proof}

\paragraph*{Properties of the middle operator $T$.}

\begin{lemma}\lb{I{T}>0}
Operator $T \colon H^{-1/2}(\Gamma)^3 \rightarrow \tilde{H}^{1/2}(\Gamma)^3$ in~(\ref{T}) is bounded and satisfies
 \beq\lb{pos-IT}
 \Im \dualGA{\bphi}{T\bphi} <0\, \qquad \forall \exs \bphi \in H^{-1/2}(\Gamma)^3:~ \bphi \neq \bzero.
 \eeq
\end{lemma}
\begin{proof}
The boundedness of $T$ stems from the well-posedness of problem~\eqref{Wik-GE} and classical trace theorems. Next, let $\bphi \in H^{-1/2}(\Gamma)^3$ and consider $\bv$ satisfying \eqref{Wik-GE} with $\bt\ff= \bphi$. Taking $\bw = \bv$ in \eqref{Wik-GE} we get 
\beq\lb{T-bound}
 \Im \dualGA{\bphi}{T\bphi} ~=~ \dualGA{\Im \bK \sip \dbv}{\llbracket\bv\rrbracket}-\Im \dualBR{\TR(\bv)}{\bv}.
\eeq
By virtue of~\eqref{T-bound}, the claim of the theorem follows immediately from Lemma~\ref{LemmaTR} and earlier hypothesis that $\Im\bK<0$. 
\end{proof}

\begin{lemma}\lb{Cp-Cr}
Operator $T \colon H^{-1/2}(\Gamma)^3 \rightarrow\tilde{H}^{1/2}(\Gamma)^3$ can be decomposed into a compact part $\Tcal_{\!c}$ and a coercive and self-adjoint part $\Tcal\zsub$ such that $T = \Tcal_{\! c} + \Tcal\zsub$. The coercive part $\Tcal\zsub \colon H^{-1/2}(\Gamma)^3 \rightarrow\tilde{H}^{1/2}(\Gamma)^3$ is defined by
\beq\lb{To}
\Tcal\zsub(\bphi) ~:=~ \llbracket \bu\Zsup\rrbracket \quad \text{on}~~ \Gamma,
\eeq 
where $\bu\Zsup \in H^1(\Bo \backslash \Gamma)$ is a solution to
\beq\lb{Wuo}
 \textrm{A}(\bu\Zsup,\bw) = \dualGA{\bphi}{\llbracket\bw \rrbracket} \quad \forall \exs \bw \in H^1(\Bo \backslash \Gamma)^3,
\eeq
$\textrm{A}$ being the coercive sesquilinear form defined by \eqref{cv}.
\end{lemma}
\begin{proof}
We first observe from~\eqref{cv} that 
\beq\lb{uoG}
\begin{aligned}
&\nabla \sip (\bC \colon \! \nabla \bu\Zsup) \,-\, \bu\Zsup ~=~ \bzero     \qquad  &\text{in} \quad \Bo \backslash \Gamma,& \\
& \bn \cdot \bC \colon \! \nxs \nabla  \bu\Zsup  ~=~ \! -\bphi    \qquad &\text{on} \quad \Gamma,& \\
&\bn \cdot \bC \colon \! \nxs \nabla\bu\Zsup ~=~ \textcolor{black}{\mathcal{S}_R(\bu\Zsup)}  \qquad &\text{on} \quad \partial{B}_R&, 
\end{aligned}
\eeq
\textcolor{black}{where $\mathcal{S}_R\!: H^{1/2}(\partial B_R)^3 \to H^{-1/2}(\partial B_R)^3$ is a Dirichlet-to-Neumann operator, $\mathcal{S}_R(\bpsi):= \bn\cdot\bC\colon\!\nabla \btu_{\bpsi}$, stemming from the \emph{elastostatic problem} in $B_{R_\circ}\!\backslash B_R$ with Dirichlet data $\btu_{\bpsi}=\bpsi$ on~$\partial B_{R}$ and homogeneous ``Neumann'' data $\bn\cdot\bC\colon\!\nabla \btu_{\bpsi}=\bzero$ on~$\partial B_{R_\circ}$.}

Using standard trace theorems for vector fields with square-integrable divergence~\cite{Monk2003}, one finds that 
\beq\lb{ToTr}
\norms{\nxs\bphi \nxs}_{H^{-\frac{1}{2}}(\Gamma)}\! \:=\: \norms{\nxs\bn\cdot\bC\colon\!\!\nabla\bu\Zsup\!}_{H^{-\frac{1}{2}}(\Gamma)}
\,\,\leqslant\: \norms{\nxs\bn\cdot\bC\colon\!\!\nabla\bu\Zsup\!}_{H^{-\frac{1}{2}}(\partial D)}
\,\,\leqslant\: c \left(\nxs\norms{\!\nabla\sip(\bC\colon\!\!\nabla \bu\Zsup)\!}_{L^2(D)} + 
\norms{\!(\bC\colon\!\!\nabla\bu\Zsup)\!}_{L^2(D)} \right)
\eeq 
for a positive constant $c$ independent from $\bu\Zsup$. Thanks to the first equation in \eqref{uoG} we then deduce 
$$
\norms{\nxs\bphi \nxs}_{H^{-\frac{1}{2}}(\Gamma)} ~\leqslant~ c_1 \norms{\bu\Zsup \!}_{H^1(D)}
$$
for some $c_1\!>\!0$ independent from $\bu\Zsup$. On taking $\bw=\bu\Zsup$ in \eqref{Wuo}, deploying the coercivity of $A$, \textcolor{black}{and recalling from~\eqref{cv} that $\Im A(\bv,\bv)=0$}, we find
\beq\lb{ToCo}
\dualGA{\bphi}{\Tcal\zsub\bphi} =  \textrm{A}(\bu\Zsup,\bu\Zsup) \geqslant c_2 \norms{\nxs\bphi \nxs}_{H^{-\frac{1}{2}}(\Gamma)}^2
\eeq
for a positive constant $c_2$ independent from $\bphi$, which establishes the coercivity of $\Tcal\zsub$. The self-adjointness of~$\Tcal\zsub$ follows immediately from that of~$\textrm{A}$.

To complete the argument, consider the compactness of $\Tcal_c \colon H^{-1/2}(\Gamma)^3 \rightarrow \tilde{H}^{1/2}(\Gamma)^3$, given by
\[
\Tcal_c(\bphi) \:=\: \llbracket\bv^c\rrbracket, \quad \bv^c \:=\: \bv - \bu\Zsup \quad \text{on}~ \Gamma
\]
where $\bv$ solves~\eqref{Wik-GE}. On subtracting~(\ref{Wuo}) from~(\ref{Wik-GE}) with~$\bt\ff=\bphi$, one finds that
\[
\textrm{A}(\bv^c, \bw) \:= - \textrm{B}(\bv,\bw) \qquad \forall \exs \bw \in H^1(\Bo \backslash \Gamma)^3,
\] 
where $\textrm{A}$ is coercive while $\text{B}$, given by~(\ref{compact}), is compact on $H^1(\Bo \backslash \Gamma)^3$. As a result, the induced mapping $\bv \rightarrow \bv^c$ from $H^1(\Bo \backslash \Gamma)^3$ into $H^1(\Bo \backslash \Gamma)^3$ is \emph{compact}, whereby the compactness of $\Tcal_c$ follows directly from the continuity of $\bv \in H^1(\Bo \backslash \Gamma)^3$ with respect to $\bphi \in H^{-\frac{1}{2}}(\Gamma)^3$ and the trace theorem.
\end{proof}

\textcolor{black}{
\begin{lemma}\lb{T-invs0}
Operator $T \colon H^{-1/2}(\Gamma)^3 \rightarrow \tilde{H}^{1/2}(\Gamma)^3$ has a bounded (and thus continuous) inverse.
\end{lemma} 
\begin{proof}
The idea is to show that $T$, given by~\eqref{T}, is injective and Fredholm of index zero. The second claim follows immediately from Lemma~\ref{Cp-Cr}.  To demonstrate the injectivity of~\eqref{T}, one may recall a double-layer potential representation of elastodynamic fields solving~\eqref{GE}-\eqref{KS} which demonstrates that for any $\bphi\in H^{-1/2}(\Gamma)$, one has  
\[
\bv(\bphi)(\bxi) ~=~ \int_{\Gamma} T(\bphi) \cdot  \bfT(\bxi,\by)  \, \text{d}S_{\by}, \qquad 
\bfT(\bxi,\by) ~=~ \bn(\by)\cdot\bSig(\bxi,\by), \qquad \bxi \in \mathbb{R}^3 \backslash \Gamma,
\]
where $\dbv=T(\bphi)$ on~$\Gamma$ thanks to the fundamental property of double-layer potentials. Thus, on assuming that there exists $\bphi\in H^{-1/2}(\Gamma)$ so that $T(\bphi)=\bzero$, one finds that $\bv\!=\!\bzero$ in~$\mathbb{R}^3\backslash\Gamma$ and consequently, by the second of~\eqref{GE} and trace theorems, that $\|\bphi\|_{H^{-1/2}(\Gamma)}=\|\bn\sip\bC\!:\!\nabla\bv\|_{H^{-1/2}(\Gamma)}=0$.  
\end{proof}}

\textcolor{black}{
\begin{lemma}\lb{T-coerc0} 
Operator $T \colon H^{-1/2}(\Gamma)^3 \rightarrow \tilde{H}^{1/2}(\Gamma)^3$ is coercive, i.e. there exists constant $c\!>\!0$ independent of~$\bphi$ such that
\beq\lb{co-T0}
|\langle \bphi, \, T (\bphi) \rangle| \,\,\geqslant\,\, \textrm{c} \nxs \norms{\bphi}_{H^{-\frac{1}{2}}(\Gamma)}^2, \qquad  \forall\bphi\in H^{-1/2}(\Gamma)^3.
\eeq 
\end{lemma}
\begin{proof}
Lemma~\ref{I{T}>0} demonstrates that the duality product $\big\langle \bphi ,\,  T(\bphi) \big\rangle\in\mathbb{C}\setminus(-\infty,\infty)$ for all nonzero $\bphi\!\in H^{-1/2}(\Gamma)^3$. Due to Lemma~\ref{Cp-Cr}, on the other hand, decomposition $T = \Tcal_{\! c} + \Tcal\zsub$ exists  where $\Tcal_{\! c}$ is compact and~$\Tcal\zsub$ is such that $\langle \exs \bphi,  \Tcal\zsub(\bphi) \exs \rangle\in \mathbb{R}$ satisfies the coercivity condition~(\ref{ToCo}) $\forall \exs \bphi\!\in\!H^{-1/2}(\Gamma)^3$. With such results in place, claim~\eqref{co-T} follows immediately by Lemma~1.17 in~\cite{Kirsch2008}. 
\end{proof}}

\section{Application of sampling methods} \label{SSA}


\subsection{Linear sampling method (LSM)}

The essential idea behind the LSM~\cite{Fiora2003} and also the factorization method (FM)~\cite{Bouk2013} for geometrical obstacle reconstruction stems from the particular nature of an approximate solution, \textcolor{black}{$\bg =\bg_p\oplus\bg_s$}, to the far-field equation 
\beq\lb{FF}
F \bg ~=~ \bPhi_L^\infty, \qquad F ~=~ \mathcal{G} \mathcal{H} ~=~ \mathcal{H}^* \exs T \exs \mathcal{H}, 
\eeq
where $\bPhi_L^\infty$ is the far-field pattern \textcolor{black}{of a trial radiating field}, see Definition~\ref{phi-infinity}. In this setting, the behavior of $\bg$ in the sampling region is exposed by characterizing the range of $\mathcal{G}$ or $\mathcal{H}^*$, \textcolor{black}{which then forms the basis for approximating the characteristic function of a scatterer}. This section presents an adaptation of the key LSM results for the problem of elastic-wave imaging of heterogeneous fractures, which provides a foundation for the GLSM developments in Section~\ref{GLSMM}.   

\begin{defn} \lb{phi-infinity}
\textcolor{black}{With reference to~\eqref{Hstar}, for every admissible FOD profile $\ba\!\in\!\tilde{H}^{1/2}(L)$ specified over a smooth, non-intersecting trial fracture $L\!\subset\!\Bo$, the induced far-field pattern $\bPhi_L^\infty \colon \tilde{H}^{1/2}(L) \rightarrow L^2(\OOd)^3$ is given by} 
\beq\lb{Phi-inf}
\begin{aligned}
\bPhi_{L}^\infty(\ba)(\hat\bxi) ~=~ -  \Big(  \, & \text{\emph{i}} k_p \,  \hat\bxi \exs \int_L \, \Big\lbrace \lambda \exs (\ba \sip \bn) \,+\, 2\mu \exs (\bn \sip \hat\bxi) ( \ba \sip \hat\bxi)  \Big\rbrace  \, e^{-\text{\emph{i}}k_p \hat\bxi \cdot \by} \,\, \text{d}S_{\by}  \\*[1 mm]
 & \textcolor{black}{\oplus}~ \text{\emph{i}} k_s \,  \hat\bxi\times\! \int_L \Big\lbrace   \mu \exs(\ba \times \hat\bxi)(\bn\sip\hat\bxi) \,+\, \mu \exs  (\bn \times \hat\bxi) (\ba \sip \hat\bxi)  \Big\rbrace \, e^{-\textrm{\emph{i}} k_s \hat\bxi \cdot \by} \,\, \text{d}S_{\by} \Big).
\end{aligned}
\eeq
and $\bn$ is the unit normal on~$L$.
\end{defn}
\begin{rem}\lb{LSMrem}
\textcolor{black}{On the basis of Definition~\ref{phi-infinity}, one may interpret the LSM reconstruction philosophy as follows. Let ${\sf L}\!\subset\!\mathbb{R}^3$ (containing the origin) denote a reference fracture surface whose characteristic size is small relative to the length scales describing the forward scattering problem, and let $L=\bz\!+\bR{\sf L}$ where $\bz\!\in\mathbb{R}^3$ and $\bR\!\in\!U(3)$ is a unitary rotation matrix. Given an admissible FOD profile $\ba\!\in\!\tilde{H}^{1/2}({\sf L})$, solving the far-field equation~\eqref{FF} over a grid of trial pairs $(\bz,\bR)$ sampling $\mathbb{R}^3\!\times U(3)$ is simply an effort to probe the far-field kernel~\eqref{w-inf} -- through synthetic rearrangement of the illuminating plane waves -- for fingerprints in terms of~$\bPhi_L^\infty$. As shown by Theorems~\ref{TR2}, \ref{GLSM1} and~\ref{GLSM2}, such fingerprint is found in the data if and only if~$L\subset\Gamma$. Otherwise, the norm of any approximate solution to~(\ref{FF}) can be made arbitrarily large, which then provides a criterion for the reconstruction of~$\Gamma$.}          
\end{rem}

\begin{theorem}\lb{TR1} 
Provided that $\omega$ is \emph{not} a ``Neumann'' eigenvalue of the Navier equation~(\ref{uiH}) and that~$\bK^{-1} \!\! \in \! L^{\infty}(\Gamma)$, for \emph{every} smooth and non-intersecting trial crack $L\subset\Bo$ and some density function $\ba(\bxi)\!\in\!\tilde{H}^{1/2}(L)$, one has
\[
\bPhi_L^\infty \in Range(\mathcal{H}^*) ~~ \iff ~~ L \subset \Gamma.
\]   
\end{theorem}
\begin{proof}
Consider the following:
\begin{itemize}
\item If~$L \subset \Gamma$, then $\tilde{H}^{1/2}(L)^3 \subset \tilde{H}^{1/2}(\Gamma)^3$. By extending the domain of $\ba\in\tilde{H}^{1/2}(L)^3$ from $L$ to $\Gamma$ through zero padding, one immediately obtains $\bPhi_L^\infty  \in Range(\mathcal{H}^*)$ thanks to~\eqref{Hstar} and~\eqref{Phi-inf}. 

\item Assume that~$L\not\subset\Gamma$ and that $\bPhi_L^\infty\!\in\!Range(\mathcal{H}^*)$. Then there exists $\bb\!\in\!\tilde{H}^{1/2}(\Gamma)^3$ such that 
\beq \notag 
\begin{aligned}
\bPhi_L^\infty(\bb)(\hat\bxi) ~=~ -\Big(& \text{\emph{i}} k_p \,  \hat\bxi \exs \int_\Gamma \, \Big\lbrace \lambda \exs (\bb \sip \bn) \,+\, 2\mu \exs (\bn \sip \hat\bxi) ( \bb \sip \hat\bxi)  \Big\rbrace  e^{-\text{\emph{i}}k_p \hat\bxi \cdot \by} \,\, \text{d}S_{\by} \\ 
 & \textcolor{black}{\oplus}~ \text{\emph{i}} k_s \,  \hat\bxi \exs \times \int_{\Gamma} \Big\lbrace   \mu \exs(\bb \times \hat\bxi)(\bn \sip \hat\bxi) \,+\, \mu \exs  (\bn \times \hat\bxi) (\bb \sip \hat\bxi)  \Big\rbrace \, e^{-\textrm{\emph{i}} k_s \hat\bxi \cdot \by} \,\, \text{d}S_{\by} \Big), 
\end{aligned}
\eeq 
 associated with the layer potential
\beq\lb{Dlpb}
\bPhi_{\exs \Gamma}(\bxi) ~=~ \int_{\Gamma} \bb(\by) \cdot  \bfT(\bxi,\by)  \, \text{d}S_{\by}, \qquad \bfT(\bxi,\by) ~=~ \bn(\by)\cdot\bSig(\bxi,\by), \qquad \bxi \in \Bo \backslash \Gamma.
\eeq
On the other hand, owing to Definition~\ref{phi-infinity} of $\bPhi_L^\infty(\hat\bxi)$, potential $\bPhi_{\exs \Gamma}(\bxi)$ should coincide with 
\beq\lb{Pb}
\bfPsi_{\! L}(\bxi) ~=~ \int_L \ba(\by) \cdot  \bfT(\bxi,\by)  \,\, \text{d}S_{\by},  \qquad \bxi \in \Bo \backslash L,
\eeq
over $ \bxi \in \Bo \backslash (L \cup \Gamma)$. Now, let $\Gamma\not\ni\bxio\!\in L$ and let $\mathcal{B}_\epsilon$ be a small ball centered at $\bxio$ such that $\mathcal{B}_\epsilon \cap \Gamma = \emptyset $. In this case $\bPhi_{\exs \Gamma}$ is analytic in $\mathcal{B}_\epsilon$, while $\bfPsi_{\! L}$ has a singularity at $\bxio\!\in\mathcal{B}_\epsilon$ -- which by contradiction completes the proof. 
\end{itemize}
\end{proof}

On the basis of the above result, one arrives at the following statement which inspires most of the LSM-based indicator functionals.    
\begin{theorem}\lb{TR2}
Under the assumptions of \textcolor{black}{Lemma~\ref{Dense-H} and~Theorem~\ref{TR1}},
\begin{itemize}
\item~If $L\!\subset\!\Gamma$, there exists a Herglotz density vector $\bg_\epsilon^L\!\in L^2(\OOd)^3$ such that $\|F\bg_\epsilon^L-\bPhi_L^\infty\|_{L^2(\OOd)} \leqslant\epsilon$ and $\limsup\limits_{\epsilon \rightarrow 0} \|\mathcal{H}\bg_\epsilon^L\|_{H^{-1/2}(\Gamma)}<\infty$.

\item~If $L \not\subset \Gamma$, then $\forall \bg_\epsilon^L\!\in L^2(\OOd)^3$ such that $\norms{\nxs F\bg_\epsilon^L-\bPhi_L^\infty  \nxs}_{L^2(\OOd)} \, \leqslant\epsilon$, one has $\,\lim\limits_{\epsilon \rightarrow 0} \norms{\mathcal{H}\bg_\epsilon^L}_{H^{-1/2}(\Gamma)} \,\,=\infty$.
\end{itemize}
\end{theorem}
\begin{proof}
\begin{description}
Let us first assume $L \subset \Gamma$, whereby $\bPhi_L^\infty \in Range(\mathcal{H}^*)$ thanks to Theorem~\ref{TR1}. Then, by definition, there exists $\ba^L\!\in\!\overline{Range(T)}$ such that $\mathcal{H}^*\ba^L = \bPhi_L^\infty$. By invoking \textcolor{black}{Lemma~\ref{T-invs0}} on the boundedness i.e.~continuity of $T^{-1}$ and Lemma~\ref{H*p} which \textcolor{black}{(by the injectivity of~$\mathcal{H}^*$)} guarantees the range denseness of $\mathcal{H}$, one finds that $\forall \epsilon>0$, $\exists \, \bg_\epsilon^L\!\in L^2(\OOd)^3$ such that $\norms{T^{-1} \ba^L\!-\!\mathcal{H}\bg_\epsilon^L}_{H^{-1/2}(\Gamma)} \,\leqslant \epsilon$. \textcolor{black}{Thanks to (i) the continuity of~$\mathcal{H}^* T$ and (ii) the fact that~$\ba^L\!\in\tilde{H}^{1/2}(\Gamma)^3$, this establishes the first part of the claim}. 
 
Next, consider the case where $L\!\not\subset\!\Gamma$ and consequently $\bPhi_L^\infty \not\in Range(\mathcal{H}^*)$ \textcolor{black}{by Theorem~\ref{TR1}}. Then, thanks to \textcolor{black}{Lemma~\ref{Dense-H} which implies the denseness of~$Range(\mathcal{H}^*)$}, \textcolor{black}{for every $\epsilon\!>\!0$ and some regularization parameter $0\!<\!\alpha\!<\!C\epsilon$ where~$C$ is a constant independent of~$\epsilon$}, a nearby solution $\ba^L_\epsilon \in \tilde{H}^{1/2}(\Gamma)^3$ can be built e.g.~via Tikhonov regularization~\cite{Kress1999} such that \textcolor{black}{$\norms{\!\bPhi_L^\infty-\mathcal{H}^* \ba_\epsilon^L\!}_{L^2(\OOd)} \,\leqslant \nxs \epsilon\,$} and $\,\lim\limits_{\epsilon \rightarrow 0} \norms{\nxs \ba^L_\epsilon \nxs}_{\tilde{H}^{1/2}(\Gamma)} \,=\infty$ -- due to the compactness of $\mathcal{H}^*$ \textcolor{black}{established in Lemma~\ref{H*p}}. At this point, the same argument as in the first part of the proof -- deploying the continuity of $T^{-1}$ and the range denseness of~$\mathcal{H}$ -- can be used to show establish the second claim.   
\end{description}
\end{proof}
\vspace{-7 mm}

\subsection{Factorization method (FM)}

To facilitate the ensuing developments, we recall elements of the factorization method~\cite{Kirsch2008} as they pertain to our inverse problem. 

\begin{defn}\lb{DFs}
The self-adjoint operator $F_\sharp \colon  L^2(\OOd)^3  \rightarrow L^2(\OOd)^3$ is defined by 
\beq\lb{Fs}
F_\sharp \,\colon \!\!\!=\, |\Re{F}| \:+\: \Im{F}, 
\vspace{-2 mm}
\eeq  
where $F\!:L^2(\OOd)^3 \to L^2(\OOd)^3$ is given by~\eqref{ffo2}, and 
\beq\lb{ReIm}
\Re{F} \,=\, \tfrac{1}{2} (F+F^*), \qquad \Im{F} \,=\, \tfrac{1}{2 \textrm{\emph{i}}} (F-F^*). 
\eeq
\end{defn}
      
\begin{rem}\lb{TsD}
In line with decomposition~(\ref{fact}) of the far-field operator $F$, there exists factorization
\beq\lb{facts}
F_\sharp ~=~ \mathcal{H}^* \exs T_\sharp \exs \mathcal{H}
\eeq
of~\eqref{Fs}, where the middle operator $T_\sharp \colon H^{-1/2}(\Gamma)^3\rightarrow\tilde{H}^{1/2}(\Gamma)^3$ is given by
\beq\lb{Tsdef}
T_\sharp :=\, \Re{T}(Q^+\!-Q^-) +\, \Im{T};
\eeq
$Q^+$ and $Q^-\!$ are bounded projectors such that $Q^{+}\!+Q^{-}=I$; $Q^{+}\!-Q^{-}$ is an isomorphism, and $Q^-\!$ has a finite rank. See Theorem 2.15 in~\cite{Kirsch2008} for derivation.
\end{rem}

\vspace*{-5mm}
\textcolor{black}{
\begin{theorem}\lb{TR3}
Under the assumptions of Theorem~\ref{TR1}, operator $F_\sharp$ in~\eqref{Fs} has the following properties: 
\begin{itemize}
\item~Operator $F_\sharp$ is positive.
\item~The ranges of~$\mathcal{H}^* \colon \tilde{H}^{1/2}(\Gamma)^3\rightarrow L^2(\OOd)^3$ and $F_\sharp^{1/2}\colon L^2(\OOd)^3\rightarrow L^2(\OOd)^3$ coincide.
\item~$\bPhi_L^\infty \in Range(F_\sharp^{1/2}) ~~ \iff ~~ L \subset \Gamma$.
\end{itemize}
\end{theorem} 
\begin{proof}
The first two claims follow directly from Theorem~2.15 in~\cite{Kirsch2008}, its extended version  (Theorem~3.2) in~\cite{Bouk2013}, Lemma~\ref{H*p} Lemma~\ref{I{T}>0}, and Lemma~\ref{Cp-Cr}. With such result in place, the last claim is immediately established by Theorem~\ref{TR1}.
\end{proof}}

\vspace*{-5mm}
\textcolor{black}{
\begin{lemma}\lb{T-invs}
Operator~$T_\sharp \colon H^{-1/2}(\Gamma)^3\rightarrow\tilde{H}^{1/2}(\Gamma)^3$ in the factorization~\eqref{facts} has the following properties:
\begin{itemize}
\item~$T_\sharp$ has a bounded (and thus continuous) inverse.
\item~$T_\sharp$ is selfadjoint and is positively coercive, i.e. there exists a constant~$c\!>\!0$ independent of~$\bphi$ so that 
\beq\lb{co-T}
\big(\bphi, \, T_\sharp (\bphi) \big)_{H^{-\frac{1}{2}}(\Gamma)} \,\,\geqslant\,\, c \nxs \norms{\bphi}_{H^{-\frac{1}{2}}(\Gamma)}^2, \qquad
\forall \,\bphi\in H^{-1/2}(\Gamma)^3. 
\eeq 
\end{itemize}
\end{lemma} 
\begin{proof}
See Appendix~A in~\cite{Bouk2013} and the proof of Theorem 2.15, part E in~\cite{Kirsch2008}. 
\end{proof}}

\textcolor{black}{On the basis of Theorem~\ref{TR2}, one sees that $F_\sharp^{1/2}$ can be used to characterize $\Gamma$ from the far-field measurements. In what follows, it is in particular shown that the GLSM cost functionals based on $F_\sharp$ (i)~are convex, (ii)~have closed-form minimizers, and (iii)~enable fast and robust reconstruction of $\Gamma$ -- especially when the data (and thus the far-field operator) are contaminated by noise.}
\subsection{Generalized Linear Sampling Method (GLSM)} \label{GLSMM}


\noindent Theorem~\ref{TR2} of the linear sampling method poses two fundamental challenges in that:~i) the featured anomaly indicator $\norms{\!\mathcal{H}\bg_\epsilon^L\!}_{H^{-1/2}(\Gamma)}$ inherently depends on the unknown fracture support $\Gamma$, and ii) construction of the Herglotz density vector $\bg_\epsilon^L \in L^2(\OOd)^3$ is implicit in the theorem~\cite{Audibert2014}. Conventionally, these issues are addressed by replacing $\norms{\!\mathcal{H}\bg_\epsilon^L\!}_{H^{-1/2}(\Gamma)}\!$ with $\norms{\nxs\bg_\epsilon^L\!}_{L^2(\OOd)}$ which is, in turn, computed by way of Tikhonov regularization~\cite{Kress1999}. Such treatment, however, has proven to be particularly sensitive to perturbations in the data due to e.g. measurement errors.    

To help meet the challenge, the GLSM takes advantage of the second factorization (\ref{fact}) of the far-field operator and the mathematical properties of its components to properly construct a \emph{stable} approximate solution to the far-field equation (\ref{FF}). This is accomplished through a \emph{sequence} of \textcolor{black}{penalized least-squares problems} where the principal ingredient of the penalty term is $\norms{\!\mathcal{H}\bg_\epsilon^L\!}_{H^{-1/2}(\Gamma)}$, reformulated in a computable way in terms of the far-field operator~$F$. More specifically, by invoking factorizations~(\ref{fact}) and~(\ref{facts}), one may observe that 
\[
\begin{aligned}
&(\bg_\epsilon^L, \exs F\bg_\epsilon^L)_{L^2(\Omega)} \,\exs ~=~ \big \langle \mathcal{H}\bg_\epsilon^L, \, T\mathcal{H}\bg_\epsilon^L \big \rangle_{\Gamma}, \\*[1 mm]
& (\bg_\epsilon^L, \exs F_\sharp \exs \bg_\epsilon^L)_{L^2(\Omega)}  ~=~ \big \langle \mathcal{H}\bg_\epsilon^L, \, T_\sharp \exs \mathcal{H}\bg_\epsilon^L \big \rangle_{\Gamma}, \qquad \forall \bg_\epsilon^L \in L^2(\OOd)^3
\end{aligned}
\]
where~$(\cdot,\cdot)_{L^2(\Omega)}:=(\cdot,\cdot)_{L^2(\Omega)^3}$ denotes the usual~$L^2$ inner product on~$\Omega$. \textcolor{black}{Then, thanks to the coercivity of the middle operator $T$ (see Lemma~\ref{T-coerc0})}, quantity $|(\bg_\epsilon^L,F\bg_\epsilon^L)_{L^2(\Omega)}|$ -- which is computable without prior knowledge of~$\Gamma$ -- may be safely substituted for $\norms{\!\mathcal{H}\bg_\epsilon^L\!}^2_{H^{-1/2}(\Gamma)}$ \textcolor{black}{in constructing a penalty term for the GLSM cost functional. Similarly, the positive coercivity $T_\sharp$ (See Lemma~\ref{T-invs}) and factorization~\eqref{facts} of $F_\sharp$ demonstrate that $|(\bg_\epsilon^L, \, F_\sharp \exs \bg_\epsilon^L)_{L^2(\OOd)}|  =  \,\, \norms{F_\sharp^{1/2}\bg_\epsilon^L}^2$ may serve as a replacement for $\norms{\!\mathcal{H}\bg_\epsilon^L\!}^2_{H^{-1/2}(\Gamma)}$}, giving birth to a \emph{convex} GLSM cost functional whose minimizer can be computed without iterations. This shines light on the GLSM approach to elastodynamic reconstruction of heterogeneous fractures, whose specificities are presented next.

\paragraph*{GLSM cost functional.}  \label{cost}
\begin{itemize}
\item~\emph{Unperturbed (noise-free) operators.}~Let $\alpha\!>\!0$ be a regularization parameter, and consider the far-field pattern~$\bPhi_L^\infty\!\in L^2(\OOd)^3$ as in Definition~\ref{phi-infinity}. Then the GLSM cost functional is defined by a sequence of penalized least-squares misfit functionals $J_\alpha(\bPhi_L^\infty;\, \cdot)\colon \, L^2(\OOd)^3 \rightarrow \mathbb{R}$, namely 
\beq \lb{J-alph}  
J_\alpha(\bPhi_L^\infty; \, \bg) ~ \colon \!\!\! =~ \! \norms{\nxs F\bg\,-\,\bPhi_L^\infty \nxs}^2 \,+\,\,\exs \alpha \nxs \norms{F_\sharp^{\frac{1}{2}} \bg}^2, \qquad \bg \in L^2(\OOd)^3, 
\eeq
whose minimizers~$\bg^L_\alpha \in L^2(\OOd)^3$ can be computed \emph{non-iteratively} by solving 
\beq \lb{min-J}
F^*(F\bg^L_\alpha \,-\, \bPhi_L^\infty)  ~+~  \alpha \exs (F_\sharp^{\frac{1}{2}} )^* F_\sharp^{\frac{1}{2}} \, \bg^L_\alpha~=~ \bzero.
\eeq  
For completeness, a more general form $\mathcal{J}_\alpha(\bPhi_L^\infty; \, \cdot)\colon \, L^2(\OOd)^3 \rightarrow \mathbb{R}$ of the GLSM cost functional, namely 
\beq \lb{fJ-alph}  
\mathcal{J}_\alpha(\bPhi_L^\infty;\, \bg) ~ \colon \!\!\! =~ \!  \norms{\nxs F\bg\,-\,\bPhi_L^\infty \nxs}^2 \,+\,\,\exs \alpha \exs | ( \bg, \exs  F \bg ) | , \qquad \bg \in L^2(\OOd)^3,
\eeq 
is also considered. \textcolor{black}{Note that~\eqref{fJ-alph} does not demand $F_\sharp$ to be applicable (see Theorem~\ref{TR3}), and thus may cater for a wider class of contact laws, $\mathcal{L}\dbv$, over the fracture surface in~(\ref{GE})}.

\begin{rem}\lb{App-sol}
In general, $\mathcal{J}_\alpha(\bPhi_L^\infty;\,\bg)$ does not have a minimizer; however, one may define 
\[
j_\alpha(\bPhi_L^\infty) ~ \colon \!\!\! = \inf\limits_{\bg \in L^2(\OOd)^3} \! \mathcal{J}_\alpha(\bPhi_L^\infty; \, \bg). 
\]
Thanks to the range denseness of $F$ (see Lemma~\ref{FF_op}), one has that $j_\alpha \rightarrow 0$ as $\alpha \rightarrow 0$. Accordingly, an optimized nearby solution can be constructed by following the algorithm described in~\cite{Audibert2014}.
\end{rem}
\item~\emph{Perturbed operators.}~When the measurements are contaminated with noise (e.g.~sensing errors, fluctuations in the medium properties), one has to deal with noisy operators $F^\delta\!$ and $F_\sharp^\delta$ satisfying  
\beq\lb{Ns-op}
\norms{\nxs F^\delta - F \nxs} \,\,\, \leqslant \,\, \delta , \qquad \norms{\nxs F^\delta_\sharp - F_\sharp \nxs} \,\,\, \leqslant \,\, \delta ,
\eeq
where $\delta\!>\!0$ is a measure of perturbation in data -- independent of $F$ and $F_\sharp$. Assuming that $F^\delta\!$ and $F_\sharp^\delta$ are compact, a regularized version $J_\alpha^\delta(\bPhi_L^\infty;\,\cdot)\colon \, L^2(\OOd)^3 \rightarrow \mathbb{R}$ of the GLSM cost functional is defined in spirit of the Tikhonov regularization method as 
\beq  \lb{RJ-alph}
J_\alpha^\delta(\bPhi_L^\infty;\, \bg) ~ \colon \!\!\! =~ \!   \norms{\nxs F^\delta\bg\,-\,\bPhi_L^\infty \nxs}^2 + \,\, \alpha \exs \big(\!\norms{\nxs(F^\delta_\sharp)^{\frac{1}{2}} \exs \bg \nxs}^2 +\,\,\exs \delta  \!  \norms{\nxs \bg \nxs}^2 \! \big), \qquad \bg\in L^2(\OOd)^3.
\eeq
Note that $J_\alpha^\delta$ is again convex and that its minimizer $\bg^L_{\alpha,\delta} \in L^2(\OOd)^3$ solves the linear system 
\beq \lb{min-RJ} 
F^{\delta *}(F^\delta \bg^L_{\alpha,\delta} \,-\, \bPhi_L^\infty) ~+~  \alpha \exs \big( \exs (F_\sharp^\delta)^{\nxs\frac{1}{2}*} (F_\sharp^\delta)^{\nxs\frac{1}{2}} \, \bg^L_{\alpha,\delta} \,+\, \delta \, \bg^L_{\alpha,\delta} \exs \big) ~=~ \bzero.
\eeq
In this vein, the (regularized) cost functional affiliated with the general form~(\ref{fJ-alph}) may be recast as
\beq\lb{RfJ-alph}
\mathcal{J}_\alpha^\delta(\bPhi_L^\infty;\,\bg) ~ \colon \!\!\! =~ \! \norms{\nxs F^\delta\bg\,-\,\bPhi_L^\infty \nxs}^2 \,+~ \alpha \exs \big( \exs | (\bg, \, F^\delta \bg ) | \,+\, \delta \! \norms{\nxs \bg \nxs}^2 \! \big), \qquad \bg\in L^2(\OOd)^3.
\eeq
\begin{rem}
\textcolor{black}{In (\ref{RJ-alph}) and (\ref{RfJ-alph}), $\delta$ signifies both a measure of perturbation in $F$ and  a regularization parameter that, along with $\alpha$, is designed to create a robust fracture indicator functional via a sequence of the GLSM minimizers (see the proof of Theorem~\ref{GLSM2}).}   
\end{rem}
\end{itemize}

With the above definitions in place, the main GLSM theorems are based on the following lemmas.

\begin{lemma}\lb{comp_G}
Operator $\mathcal{G}=\mathcal{H}^* T \colon H^{-1/2}(\Gamma)^3 \rightarrow  L^2(\OOd)^3$ is compact over $H^{-1/2}(\Gamma)^3$. 
\end{lemma}
\begin{proof}
The claim follows immediately from Lemmas~\ref{H*p} and~\ref{I{T}>0} establishing, respectively, the compactness of $\mathcal{H}^*$ and the boundedness of $T$.
 \end{proof}
 \begin{lemma}\lb{FF_op}
The far-field operator $F \colon   L^2(\OOd)^3  \rightarrow  L^2(\OOd)^3$ is injective, compact and, under the assumptions of Lemma~\ref{Dense-H}, has a dense range.  
\end{lemma}
\begin{proof}
\emph{Injectivity.}~Let $F(\bg) = \bzero$. \textcolor{black} {Then, recalling the factorization $F = \mathcal{H}^* T \mathcal{H}$ and the injectivity of $\mathcal{H}^*$ and~$T$  (due respectively to Lemma~\ref{Dense-H} and Lemma~\ref{T-invs0}), one finds that $\mathcal{H}(\bg):= \bn\cdot\bC\colon\!\nabla\pff_\bg=\bzero$ on~$\Gamma$. Under the assumptions of Lemma~\ref{Dense-H}, this requires that $\pff_\bg=\bzero$ in~$\mathbb{R}^3$, i.e. that~$\bg=\bzero$ which establishes the first claim.} 

\emph{Compactness.}~ The compactness of~$F$ follows immediately from the compactness of  $\mathcal{H}^*$ -- and thus that of $\mathcal{H}$ (Lemma~\ref{H*p}), and the boundedness of~$T$ (Lemma~\ref{I{T}>0}). 

\emph{Range densenes.}~\textcolor{black}{This claim is conveniently verified  by establishing the injectivity of~$F^*$. To this end,  recall~(\ref{ffo2}) and consider the $L^2$-inner product
\beq\lb{Dp-F*}
\big(F(\bg),\ba\big)_{L^2(\Omega)} ~=~ \int_{\Omega} \bar\ba(\hat\bxi) \cdot \bv_{\bg_\Omega}^\infty(\hat\bxi) \, \text{d}S_{\hat\bxi} ~=~ \int_{\OOd} \bg(\bd) \cdot \overline{\int_{\Omega} \bW^{\infty*}(\bd,\hat\bxi)\cdot \ba(\hat\bxi) \, \text{d}S_{\hat\bxi}} \,\, \text{d}S_{\bd},
\eeq
where~$\ba\in L^2(\Omega)^3$. Thanks to the reciprocity identity~(\ref{W-recip}), inner product~\eqref{Dp-F*} exposes the adjoint far-field operator as 
\beq\lb{F*}
 F^*(\ba)(\bd) ~=~ \int_{\Omega} \bW^{\infty*}(\bd,\hat\bxi)\cdot \ba(\hat\bxi) \, \text{d}S_{\hat\bxi} ~=~
\overline{\int_{\Omega} \bW^{\infty}(\hat\bxi,-\bd)\cdot \overline\ba(-\hat\bxi) \, \text{d}S_{\hat\bxi}} ~=~ \overline{F}(\tilde{\ba})(-\bd), \quad~ \bd\in\Omega,
\eeq
where $\tilde\ba(\hat\bxi)\!:=\!\overline{\ba}(-\hat\bxi)$ on~$\Omega$. Owing to the injectivity of $F$, one finds from~\eqref{F*} that setting $F^*(\ba)\!=\!\bzero$ necessitates~$\tilde\ba=\bzero$ and thus $\ba=\bzero$.}
\end{proof}

We are now in position to establish the main result of the GLSM approach, given by Theorem~\ref{GLSM1} and Theorem~\ref{GLSM2}, catering for the elastodynamic reconstruction of heterogeneous fractures.

\begin{theorem} \lb{GLSM1}
\textcolor{black}{Consider the GLSM cost functional $\mathfrak{J}_\alpha$ unifying~\eqref{J-alph} and~\eqref{fJ-alph} with unperturbed operators $F^\delta$ and~$F_\sharp^\delta$, namely} 
\beq\lb{GCf}
\mathfrak{J}_\alpha(\bPhi_L^\infty;\,\bg) ~:=~ \norms{\nxs F\bg\,-\,\bPhi_L^\infty \nxs}_{L^2(\Omega)}^2 \,+~ \alpha\,|(\bg, B\bg)|, \qquad \bg\in L^2(\OOd)^3,
\eeq
\textcolor{black}{where $\alpha>0$ and~$B$, denoting either~$F$ or~$F_\sharp$, admits the factorization}  
\beq\lb{Bdf}
B ~=~ \mathcal{H^*} \exs \mathfrak{T} \exs \mathcal{H}, \qquad \mathfrak{T} ~=~ T,\,T_\sharp.
\eeq
Since  $\mathfrak{J}_\alpha\geqslant 0$, define the infimum 
\[
\mathfrak{j}_\alpha(\bPhi_L^\infty) ~\colon\!\!=~\! \inf\limits_{\bg \in L^2(\OOd)^3} \! \mathfrak{J}_\alpha(\bPhi_L^\infty;\bg), 
\]
and let $\bg_\alpha^L \in L^2(\OOd)^3$ denote a nearby solution such that
\[
\mathfrak{J}_{\alpha}(\bPhi_L^\infty;\bg_\alpha^L) \,\,\leqslant \,\,  \mathfrak{j}_\alpha(\bPhi_L^\infty) + \mu \alpha, 
\]
$\mu>0$ being a constant independent of $\alpha$. Then,
\[
 \bPhi_L^\infty \in Range(\mathcal{H}^*) ~\iff~ \Big\{\limsup\limits_{\alpha \rightarrow 0} |( \bg_\alpha^L,  B \bg_\alpha^L ) | < \infty ~\iff~ \liminf\limits_{\alpha \rightarrow 0} |( \bg_\alpha^L,  B \bg_\alpha^L ) | < \infty\Big\}.
 \]
\end{theorem}
\begin{proof} See the proof of Theorem~3 in~\cite{Audibert2014}, synthesized in~\ref{GLSM*pruf} using present notation. 
\end{proof}

\begin{lemma}\lb{min-Jad}
Consider the regularized GLSM cost functional $\mathfrak{J}_\alpha^\delta$ unifying~\eqref{RJ-alph} and~\eqref{RfJ-alph} with perturbed operators $F^\delta$ and~$F_\sharp^\delta$, namely  
\beq\lb{GCfn}
\mathfrak{J}_\alpha^\delta(\bPhi_L^\infty;\bg) ~:=~  \norms{\nxs F^\delta\bg\,-\,\bPhi_L^\infty \nxs}^2 \,+~ \alpha \exs \big( \exs | ( \bg, B^\delta \bg ) | \,+\, \delta \! \norms{\nxs \bg \nxs}^2 \! \big), \qquad \bg \in L^2(\OOd)^3 
\eeq
where $\alpha,\beta\!>\!0$ and $B^\delta$ denotes either~$F^\delta$ or~$F^\delta_\sharp$.  Assuming that~$B^\delta$ is compact, $\mathfrak{J}_\alpha^\delta$ has a minimizer $\bg_{\alpha,\delta}^L \in L^2(\OOd)^3$ satisfying
\beq\lb{limlim}
\lim\limits_{\alpha \rightarrow 0} \limsup\limits_{\delta \rightarrow 0} \exs \mathfrak{J}_\alpha^\delta(\bPhi_L^\infty;\,\bg_{\alpha,\delta}^L) ~=~ 0.
\eeq
\end{lemma}
\begin{proof}
\emph{Existence of a minimizer.}~For any $\alpha, \delta > 0$ and $\bPhi_L^\infty\!\in L^2(\OOd)^3$ given by~\eqref{Phi-inf}, any sequence $(\bg^n)$ constructed to minimize $\mathfrak{J}_\alpha^\delta$ is bounded in $L^2(\OOd)^3$, and thus weakly convergent to some $\bg_{\alpha,\delta}^L \in L^2(\OOd)^3$. Thanks to the lower semi-continuity of a norm with respect to the weak convergence and the postulated compactness of $B^\delta$, one has
\beq\lb{ming}
\mathfrak{J}_\alpha^\delta(\bPhi_L^\infty;\bg_{\alpha,\delta}^L) \,\, \leqslant \,\, \liminf\limits_{n \rightarrow \infty} \exs \mathfrak{J}_\alpha^\delta(\bPhi_L^\infty;\bg^n) \,\, \leqslant \,\, \!\! \inf\limits_{\bg \in L^2(\OOd)^3} \! \mathfrak{J}_\alpha^\delta(\bPhi_L^\infty;\bg),
\eeq
which proves that $\bg_{\alpha,\delta}^L$ is a minimizer of $\mathfrak{J}_\alpha^\delta(\bPhi_L^\infty;\bg)$ in~$L^2(\OOd)^3$. 

\emph{Limiting behavior.} \textcolor{black}{ Let us first observe from~\eqref{Ns-op}, \eqref{GCf} and~\eqref{GCfn} that
\beq\lb{JadJa}
\mathfrak{J}_\alpha^\delta(\bPhi_L^\infty;\bg) \,\,\leqslant \,\, \mathfrak{J}_\alpha(\bPhi_L^\infty;\bg) ~+~ 
\delta \big\{2 \alpha\|\bg\|^2 \,+\, \delta\|\bg\|^2  \,+\, 2\|F\bg-\bPhi_L^\infty\|\exs \|\bg\|\big\}, \qquad \forall \bg \in L^2(\OOd)^3.
\eeq
For any $\delta\!>\!0$ ($\alpha$ fixed), on can chose  $\bg_{\alpha,\delta}$ such that $|\mathfrak{J}_\alpha(\bPhi_L^\infty;\bg_{\alpha,\delta})- \mathfrak{j}_\alpha(\bPhi_L^\infty)| \, \leqslant \, \delta$. Then by the definition of~$\bg_{\alpha,\delta}^L$ one finds via triangle inequality that
\[
\mathfrak{J}_\alpha^\delta(\bPhi_L^\infty;\bg_{\alpha,\delta}^L) \, \leqslant \
\mathfrak{J}_\alpha^\delta(\bPhi_L^\infty;\bg_{\alpha,\delta}) \, \leqslant \, \mathfrak{j}_\alpha(\bPhi_L^\infty) + 
\delta \big\{1+ 2 \alpha\|\bg_{\alpha,\delta}\|^2 \,+\, \delta\|\bg_{\alpha,\delta}\|^2  \,+\, 2\|F\bg_{\alpha,\delta}-\bPhi_L^\infty\|\exs \|\bg_{\alpha,\delta}\|\big\}. 
\]
The proof of~\eqref{limlim} is now completed by noting that (i) given~$\alpha$, the term inside the brackets is bounded for any~$\delta$, and (ii) $\lim\limits_{\alpha\to 0} \mathfrak{j}_\alpha=0$. }
\end{proof}

\begin{theorem}\lb{GLSM2}
Under the assumptions of Theorem~\ref{GLSM1} and an additional hypothesis that $B^\delta$ (denoting either~$F^\delta$ or~$F^\delta_\sharp$) is compact, one has 
\begin{multline}\notag
 \bPhi_L^\infty \,\in\, Range(\mathcal{H}^*)  ~~\iff~~ 
\Big\{\limsup\limits_{\alpha \rightarrow 0}\limsup\limits_{\delta \rightarrow 0}\big(\exs |(\bg_{\alpha,\delta}^L, B^\delta \bg_{\alpha,\delta}^L)| \,+\, \delta \! \norms{\! \bg_{\alpha,\delta}^L \!}^2 \nxs \big) \,<\, \infty  \\  
\iff~ \liminf\limits_{\alpha \rightarrow 0}\liminf\limits_{\delta \rightarrow 0}\big(\exs |( \bg_{\alpha,\delta}^L, B^\delta \bg_{\alpha,\delta}^L)| \,+\, \delta \! \norms{\! \bg_{\alpha,\delta}^L \!}^2 \nxs \big) \,<\, \infty\Big\}, 
\end{multline}
where~$\bg_{\alpha,\delta}^L$ is a minimizer of the perturbed GLSM cost functional~\eqref{GCfn} in the sense of~\eqref{limlim}. 
\end{theorem}
\begin{proof} See the proof of Theorem~5 in~\cite{Audibert2014}, also summarized in~\ref{GLSM*pruf}. 
\end{proof}
\vspace{-7 mm}
\subsubsection{The GLSM criteria for imaging heterogeneous fractures} \label{GLSM_C}

On the basis of Theorem~\ref{GLSM2}, a robust GLSM-based criterion for the elastic-wave reconstruction of heterogeneous fractures can be designed as
\beq\lb{GLSMg}
I^{\mathcal{G}}(L) \,\, \colon \!\!\! = \,\, \dfrac{1}{\sqrt{|( \bg_{\alpha,\delta}^L, B^\delta \bg_{\alpha,\delta}^L)| \exs+\exs \delta \! \norms{ \bg_{\alpha,\delta}^L \!}^2}}, \qquad  B^\delta ~=~ F^\delta, F^\delta_\sharp,
\eeq
where $\bg_{\alpha,\delta}^L$ is a minimizer of~\eqref{GCfn} in the sense of~\eqref{limlim}. In this setting, it is particularly instructive to focus on the case where $B^\delta = F^\delta_\sharp$, since $\bg_{\alpha,\delta}^L$ in this case can be obtained \emph{non-iteratively} by explicitly solving~(\ref{min-RJ}). Accordingly, the GLSM indicator functional used is the sequel is taken as 
\beq\lb{GLSMgs}
I^{{\mathcal{G}}_\sharp}(L) \,\, = \,\, \dfrac{1}{\sqrt{\norms{\!(F^\delta_\sharp)^{\frac{1}{2}} \exs \bg_{\alpha,\delta}^L \nxs}^2 \exs+\,\, \delta \! \norms{ \bg_{\alpha,\delta}^L \!}^2}}. 
\eeq      
For future reference, let us also recall the classical LSM/FM solution $\bg_\epsilon^L \in L^2(\OOd)^3$ (see Theorem~\ref{TR2}) obtained by way of Tikhonov regularization~\cite{Kress1999}, namely 
\beq\lb{LSMg}
\textcolor{black}{
\bg_\epsilon^L   \,\, \colon \!\!\! = \,\, \min_{\bg \in L^2(\OOd)^3}  \big\lbrace \! \norms{F^\delta \bg \,-\, \bPhi_L^\infty}^2 + \,\, \beta \! \norms{\bg}^2 \! \big\rbrace,} 
\eeq
where~$\beta$ is a regularization parameter computable by the Morozov discrepancy principle.

\begin{rem}
It is worth noting that the GLSM characterization of $\Gamma$ from the far-field data (via the range of $F$) is deeply rooted in geometrical considerations, so that the fracture indicator functionals~(\ref{GLSMg}) and~(\ref{GLSMgs}) may exhibit only a minor dependence on its heterogeneous contact condition -- given by the distribution of~$\bK$ on~$\Gamma$. This behavior can be traced back to Remark~\ref{LSMrem}, where the opening displacement profile $\ba \in \tilde{H}^{1/2}(L)$ -- intimately related to the interface law -- is deemed arbitrary (within the constraints of admissibility). This quality makes the GLSM imaging paradigm particularly attractive in situations where the fracture's contact law is unknown beforehand, which opens up possibilities  for the sequential geometrical reconstruction and interfacial characterization of partially-closed fractures.   
\end{rem}


\section{Computational treatment and results} \label{numerics}

To illustrate the theoretical developments, this section examines the performance of~\eqref{GLSMgs} through a set of numerical experiments and compares the results of the GLSM reconstruction to those obtained by two alternative approaches, namely the linear sampling method (LSM)~\cite{Fiora2003} and the method of topological sensitivity (TS)~\cite{Fatemeh2015}. In what follows the synthetic sensory data, namely the far-field patterns~\eqref{vinf2} over the unit sphere, are generated by way of an elastodynamic boundary integral method~\cite{Fatemeh2015}. 
\begin{figure}[tp]
\center\includegraphics[width=0.88\linewidth]{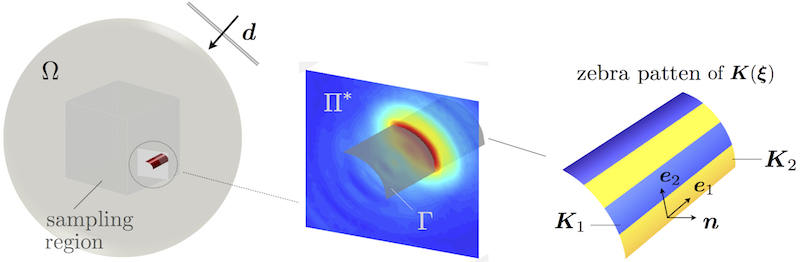} \vspace*{0mm} 
\caption{Elastic-wave sensing setup (left), position of the cutting plane (middle), and ``zebra'' pattern of the fracture's heterogeneous contact condition~(right).} \lb{SetNum}
\end{figure} 

\emph{Testing configuration.}~The sensing setup, shown in Fig.~\ref{SetNum}, features a ``true'' cylindrical fracture $\Gamma$ of length $L = 0.7$ and radius $R = 0.35$. The fracture is endowed with a piecewise-constant (``zebra'') distribution of interfacial stiffness~$\bK(\bxi)$ on~$\Gamma$, alternating between $\bK_1$ and~$\bK_2$, where 
\[
\begin{aligned}
\bK_1 \:=\: (1-0.25\textrm{i}) \, \bn \otimes \bn \,\,+\, (4 - 2\textrm{i}) \, \be_1 \nxs \otimes \be_1 \,+\, (4 - 2\textrm{i}) \, \be_2 \nxs \otimes \be_2, \qquad \bK_2 \:=\: \bzero 
\end{aligned}
\]
in terms of the orthonormal basis~$(\be_1, \be_2,\bn)$ shown in the figure. The shear modulus, mass density, and Poisson's ratio of the background solid are taken as $\mu = 1$, $\rho = 1$ and $\nu = 0.35$, whereby the shear and compressional wave speeds read $c_s = 1$ and $c_p = 2.08$, respectively. The interaction of $\Gamma$ with incident (P- and S-) plane waves, propagating in direction $\bd$, gives rise to the scattered wavefield $\bv$ solving~\eqref{GE} -- whose far-field pattern $\bv^\infty$ is then computed on the basis of~(\ref{vinf2}). 

\emph{Far-field operator.}~For both illumination and sensing purposes, the unit sphere $\Omega$ is sampled by a uniform grid of $N_\theta \!\times\! N_\phi$ observation directions, specified by the polar ($\theta_j,\, j\!=\!1,\ldots N_\theta$) and azimuthal~($\phi_k, \,k\!=\!1,\ldots N_\phi$) angle values. With reference to~\eqref{mat1}, note that both the polarization vector $\bq\!=\!\bq_p\!\oplus\bq_s$ of an incident plane wave and the far-field pattern $\bv^\infty_\bq=\bv^\infty_{\bq_p}\!\oplus \bv^\infty_{\bq_s}$ of the scattered wave each have \emph{only three} nontrivial components. In this setting, the discretized far-field operator $\textrm{\bf{F}}$ is represented as a $3N\!\times 3N$ matrix ($N\!=\!N_\theta N_\phi$) with components   
\beq\lb{DF}
\textrm{\bf{F}}(3k\nxs+\nxs1\!:\!3k\nxs+\nxs3, \,3j\nxs+\nxs1\!:\!3j\nxs+\nxs3) ~=~ \textrm{\bf{W}}^\infty (\bd_j,\hat\bxi_k), \qquad j,k = 0,\ldots N-1,
\eeq   
where
\beq\lb{mat2}
\textrm{\bf{W}}^\infty (\bd_j,\hat\bxi_k) ~=~ 
\left[\begin{array}{cccccc}
W^{\infty}_{11} & W^{\infty}_{12} & W^{\infty}_{13} \\
W^{\infty}_{21} & W^{\infty}_{22} & W^{\infty}_{23} \\
W^{\infty}_{31} & W^{\infty}_{32} & W^{\infty}_{33} 
\end{array}\right] (\bd_j,\hat\bxi_k), 
\eeq
and~$W^{\infty}_{kj}$ $(j,k\!=\!1,2,3)$ are specified in~\eqref{mat1}. Unless stated otherwise, we assume $N_\theta=50$ and $N_\phi=25$. 

\emph{Noisy data.}~To account for the presence of noise in measurements, we consider the perturbed far-field operator    
\beq\lb{DFN}
\textrm{\bf{F}}^\delta \,\, \colon \!\!\!= \, (\boldsymbol{I} + \boldsymbol{N}_{\!\epsilon} ) \exs \textrm{\bf{F}},
\eeq
where $\boldsymbol{I}$ is the $3N \times 3N$ identity matrix, and $\boldsymbol{N}_{\!\epsilon}$ is the noise matrix of commensurate dimension whose components are uniformly-distributed (complex) random variables in $[-\epsilon, \, \epsilon]^2$. On the basis of definition~(\ref{Ns-op}), one has $\delta = \norms{\!\boldsymbol{N}_{\!\epsilon} \exs \textrm{\bf{F}}\!}$ which in the sequel takes values of up to~$20\%$. With reference to Remark~\ref{LSMrem}, the region of interest 

\emph{Trial far-field pattern.}  With reference to Remark~\ref{LSMrem}, the GLSM indicator map~\eqref{GLSMgs} is constructed by solving~\eqref{min-RJ} for the minimizer of~(\ref{RJ-alph}) over a grid of trial infinitesimal fractures $L=\bz\!+\bR{\sf L}$, where~$\bz$ denotes the sampling point and $\bR$ is a unitary rotation matrix. In what follows, this is accomplished by taking~{\sf L} to be a vanishing penny-shaped fracture with unit normal~$\bn_\circ$, i.e. by setting the FOD in~(\ref{Phi-inf}) as $\ba(\by) = \delta (\by-\bz) \bR\bn_\circ$. Writing for brevity $\textrm{\bf{n}}=\bR\bn_\circ$, one in particular finds that    
\beq\lb{Phi-inf-num}
\bPhi_L^\infty(\hat\bxi) ~=~ -  \Big(\text{i} k_p \,  \hat\bxi \exs  \big[ \exs  \lambda+2\mu \exs (\textrm{\bf{n}} \cdot \hat\bxi)^2  \exs \big] \exs e^{-\text{i}k_p \hat\bxi \cdot \bz} \;\oplus\;
 2 \text{i} \mu \exs k_s \,\hat\bxi \times \nxs (\textrm{\bf{n}} \times\hat\bxi)\exs   (\textrm{\bf{n}}\cdot\hat\bxi) \, e^{-\textrm{\emph{i}} k_s \hat\bxi \cdot \bz}  \Big).
\eeq   
Recalling~\eqref{mat1}, one may note that for each observation direction~$\hat\bxi_k$, \eqref{Phi-inf-num} has only three non-trivial components in the reference $(\hat\bxi_k,\boldsymbol{\theta}_k,\boldsymbol{\phi}_k)$ orthonormal basis, which are for consistency with~\eqref{mat2} arranged as a $3N\!\times\!1$ vector 
\beq\lb{Phi-inf-Dnum}
\bPhi_{\bz,\textrm{\bf{n}}}^\infty(3k+1\!:\!3k+3) ~=\, 
\left[\begin{array}{c} \text{i} k_p  \big[ \exs  \lambda + 2\mu \exs (\textrm{\bf{n}} \sip \hat\bxi_k)^2  \exs \big] e^{-\text{i}k_p \hat\bxi_k \sip \bz} \\
2 \text{i} \mu \exs k_s (\textrm{\bf{n}}\sip\boldsymbol{\theta}_k) (\textrm{\bf{n}}\sip\hat\bxi_k) \, e^{-\textrm{\emph{i}} k_s \hat\bxi_k \sip \bz} \\ 
2 \text{i} \mu \exs k_s  (\textrm{\bf{n}}\sip\boldsymbol{\phi}_k) (\textrm{\bf{n}}\sip\hat\bxi_k) \, e^{-\textrm{\emph{i}} k_s \hat\bxi_k \sip \bz}
\end{array} \right], \qquad k=0,\ldots N-1.
\eeq
Accordingly, the far-field equation~\eqref{FF} takes the discretized form 
\beq\lb{Dff}
\textrm{\bf{F}}^\delta \bg_{\bz,\textrm{\bf{n}}} ~=~ \bPhi_{\bz,\textrm{\bf{n}}}^\infty, 
\eeq
thus forming the basis for computing GLSM and LSM indicator functionals. 

\subsection{Fracture indicators} \label{RRM}
As shown in Fig.~\ref{SetNum}, the search area i.e.~the sampling region is a \emph{cube of side 2} where the featured (GLSM and LSM) indicator functionals are evaluated. The resulting distributions are plotted either in three dimensions, or in the mid-section of the ``true'' cylindrical fracture (see Fig.~\ref{SetNum}). 

\emph{Sampling.}~In what follows, the search cube $[-1,1]^3\subset\mathbb{R}^3$ is probed by a uniform $40 \!\times\! 40 \!\times\! 40$ grid of sampling points~$\bz$, while the unit sphere -- spanning possible fracture orientations -- is sampled by a $24 \!\times\! 6$ grid of trial normal directions $\textrm{\bf{n}}=\bR\bn_\circ$. Accordingly, the fracture indicator map is constructed by solving~(\ref{Dff}) for a total of $M = 64000 \!\times\! 144$ trial pairs $(\bz,\textrm{\bf{n}})$.     

\emph{GLSM indicator.}~With reference to~\eqref{min-RJ} and~\eqref{DF}-\eqref{Dff}, a discretized version of the GLSM solution vector, $\bg^{\mbox{\tiny{GLSM}}}_{\bz,\textrm{\bf{n}}}$, is computed by solving the linear system 
\beq \lb{min-DRJ} 
\Big( \textrm{\bf{F}}^{\delta *}\textrm{\bf{F}}^\delta  + \alpha_{\bz,\textrm{\bf{n}}} \exs  (\textrm{\bf{F}}_\sharp^\delta)^{\nxs\frac{1}{2}*} (\textrm{\bf{F}}_\sharp^\delta)^{\nxs\frac{1}{2}} + \alpha_{\bz,\textrm{\bf{n}}} \delta \exs \boldsymbol{I} \Big) \exs \bg^{\mbox{\tiny{GLSM}}}_{\bz,\textrm{\bf{n}}}  ~=~  \textrm{\bf{F}}^{\delta *} \bPhi^\infty_{\bz,\textrm{\bf{n}}},
\eeq
where $(\cdot)^*$ is the Hermitian operator; $\textrm{\bf{F}}_\sharp^\delta$ is evaluated on the basis of definitions~(\ref{Fs}) and~(\ref{ReIm}); and, following~\cite{Audibert2014},     
\beq\lb{Alph}
\alpha_{\bz,\textrm{\bf{n}}} \,\, \colon \!\!\! = \,\, \frac{\eta_{\bz,\textrm{\bf{n}}}}{\norms{\textrm{\bf{F}}^\delta\!} + \,\, \delta}.
\eeq
Here $\eta_{\bz,\textrm{\bf{n}}}$ is a regularization parameter of the classical LSM solution~\eqref{lssm1}, computed via the Morozov discrepancy principle~\cite{Kress1999}. With reference to~(\ref{GLSMgs}), the GLSM indicator function is then obtained as 
\beq\lb{GLSM-Dgs}
I^{\mathcal{G}_\sharp}(\bz) \,\, = \,\, \dfrac{1}{\sqrt{\norms{\!(\textrm{\bf{F}}^\delta_\sharp)^{\frac{1}{2}} \exs \bg^{\mbox{\tiny{GLSM}}}_{\bz} \nxs}^2 \exs+\,\, \delta \! \norms{\bg^{\mbox{\tiny{GLSM}}}_{\bz} \!}^2}}, \qquad 
\textcolor{black}{
\bg^{\mbox{\tiny{GLSM}}}_{\bz} \,\,\colon \!\!= \,\, \text{argmin}_{\bg^{\mbox{\tiny{GLSM}}}_{\bz, \textrm{\bf{n}}}} \norms{\bg^{\mbox{\tiny{GLSM}}}_{\bz,\textrm{\bf{n}}}}^2_{L^2(\Omega)}, ~ \textrm{\bf{n}}\in\Omega.} 
\eeq  

\emph{LSM indicator.}~To gain better insight into the effectiveness of the proposed approach, the GLSM reconstruction is compared to a corresponding~LSM map. The latter is computed on the basis of a Tikhonov-regularized solution $\bg^{\mbox{\tiny{LSM}}}_{\bz,\textrm{\bf{n}}}$ to~(\ref{Dff}), namely
\beq\label{lssm1}
\bg^{\mbox{\tiny{LSM}}}_{\bz,\textrm{\bf{n}}} \,\,\colon \!\!= \,\, \text{argmin}_{\bg_{\bz, \textrm{\bf{n}}}} \Big\{  \norms{\textrm{\bf{F}}^\delta \bg_{\bz,\textrm{\bf{n}}} \,-\, \bPhi_{\bz,\textrm{\bf{n}}}^\infty}^2_{L^2(\Omega)} \,+\,\,\,  \eta_{\bz,\textrm{\bf{n}}} \norms{\bg_{\bz,\textrm{\bf{n}}}}^2_{L^2(\Omega)}\Big\},
\eeq
where the regularization parameter $\eta_{\bz,\textrm{\bf{n}}}$ is obtained by way of Morozov discrepancy principle~\cite{Kress1999}. On the basis of~\eqref{lssm1}, the LSM indicator functional is constructed following~\cite{Fiora2003} as 
\beq\lb{LSM}
I^{\mathcal{L}}(\bz) \,\, := \,\, \frac{1}{\norms{\bg^{\mbox{\tiny{LSM}}}_{\bz}}^2}, \qquad
\textcolor{black}{
\bg^{\mbox{\tiny{LSM}}}_{\bz} \,\,\colon \!\!= \,\, \text{argmin}_{\bg^{\mbox{\tiny{LSM}}}_{\bz, \textrm{\bf{n}}}} \norms{\bg^{\mbox{\tiny{LSM}}}_{\bz,\textrm{\bf{n}}}}^2_{L^2(\Omega)}, ~ \textrm{\bf{n}}\in\Omega.}
\eeq


\subsection{Results} \label{Comp}

In the sequel, the arclength ($\ell \!=\! 0.55$) of a ``true'' cylindrical fracture in its mid-plane, see Fig.~\ref{SetNum}, is used as a reference length to gauge the illuminating shear wavelength $\lambda_s = 2\pi/k_s$. 

\emph{Density of the sensing grid}. Taking $\lambda_s/\ell = 0.7$, Fig.~\ref{NN} illustrates the sensitivity of the GLSM indicator~\eqref{GLSM-Dgs} to the spatial density of sensory data, given by $N_\theta \!\times\! N_\phi$ incident/observation directions over the unit sphere. This is done by gradual downsampling of the default $50 \!\times\! 25$ sensing grid. From the panels, it is apparent that for satisfactory geometric reconstruction, the sensing grid should carry at least 100 test directions over $\Omega$. In what follows, the (full-aperture) reconstructions are implemented using a $50 \!\times\! 25$ grid.
 
\emph{Sensitivity to measurement noise}. Assuming full-aperture illumination and sensing, the GLSM and LSM indicators are next compared in terms of their robustness against noise in the far-field data. With reference to~\eqref{DFN}, the levels of ``white'' noise used to contaminate the boundary integral simulations of the forward scattering problem are taken $\delta = \norms{\!\boldsymbol{N}_{\!\epsilon} \exs \textrm{\bf{F}}\!} \in\{0, 0.1, 0.2\} \| \textrm{\bf{F}}\|$. On focusing the comparison on the mid-section~$\Pi$ of a ``true'' fracture, the results are shown in Figs.~\ref{F1}, \ref{F2}, and \ref{F4} assuming the illuminating wavelengths of~$\lambda_s/\ell = 1.3$, $0.7$, and $0.3$, respectively. Note that $\delta\%:=\delta/\|\textrm{\bf{F}}\|$. As can be seen from the display, the GLSM indicator~\eqref{GLSM-Dgs} inherits the superior localization ability of its LSM predecessor~\eqref{LSM}, while carrying far greater robustness to noise in the sensory data. 
\begin{figure}[h]
\center\includegraphics[width=0.64\linewidth]{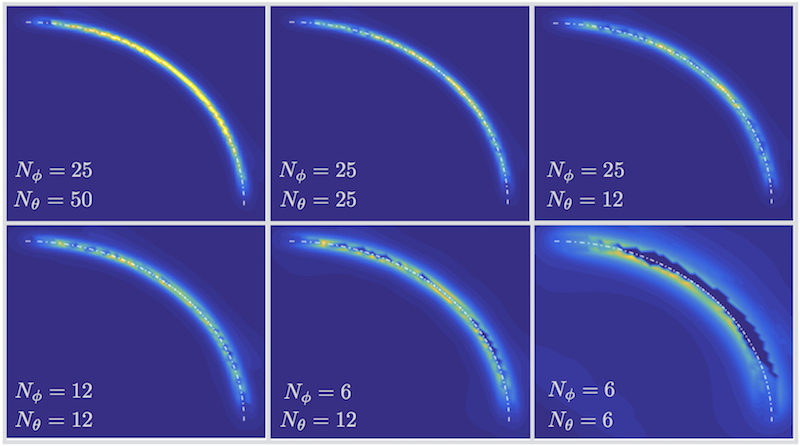} \vspace*{0mm} 
\caption{Full-aperture GLSM reconstruction of a cylindrical fracture in its mid-section, $\Pi$, for $\lambda_s/\ell = 0.7$: effect of density of the $N_\theta \!\times\! N_\phi$ sensing grid of illumination/observation directions spanning the unit sphere.} \lb{NN}
\end{figure}

\begin{figure}[h!]
\center\includegraphics[width=0.66\linewidth]{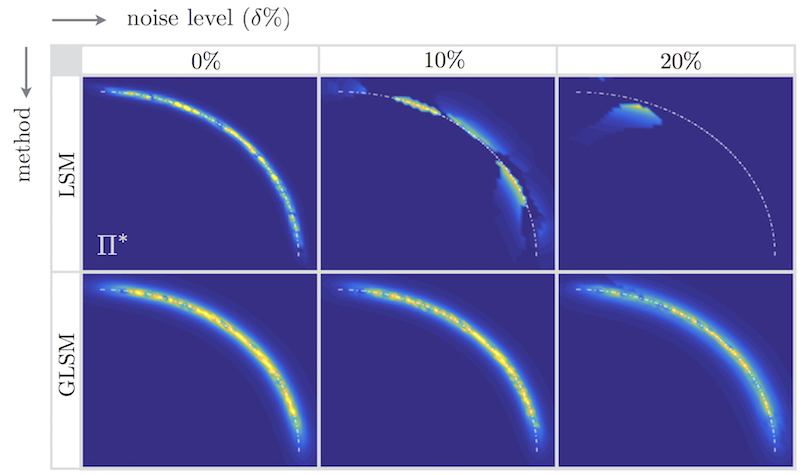} \vspace*{-2mm} 
\caption{Sensitivity to measurement noise for $\lambda_s/\ell = 1.3$: Full-aperture reconstruction of a cylindrical fracture, mid-section~$\Pi$, by the LSM indicator (top panels) and its GLSM counterpart (bottom panels).} \lb{F1}  \vspace*{2mm}
\center\includegraphics[width=0.66\linewidth]{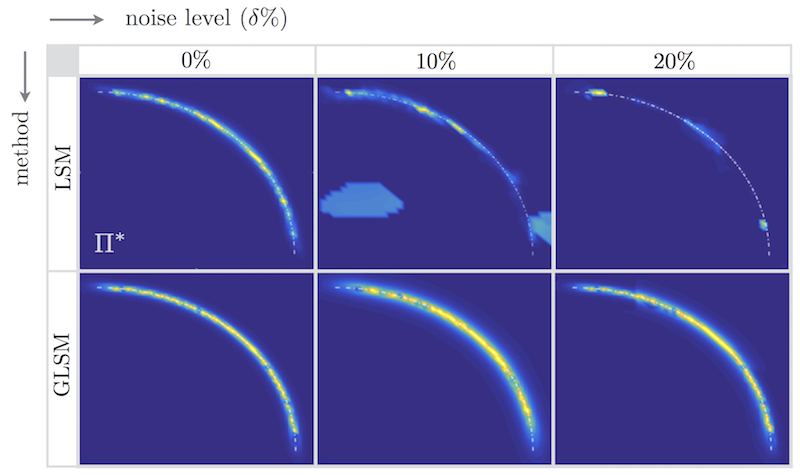} \vspace*{-2mm} 
\caption{Sensitivity to measurement noise for $\lambda_s/\ell = 0.7$: Full-aperture reconstruction of a cylindrical fracture, mid-section~$\Pi$, by the LSM indicator (top panels) and its GLSM counterpart (bottom panels).} \lb{F2} \vspace*{2mm}
\center\includegraphics[width=0.66\linewidth]{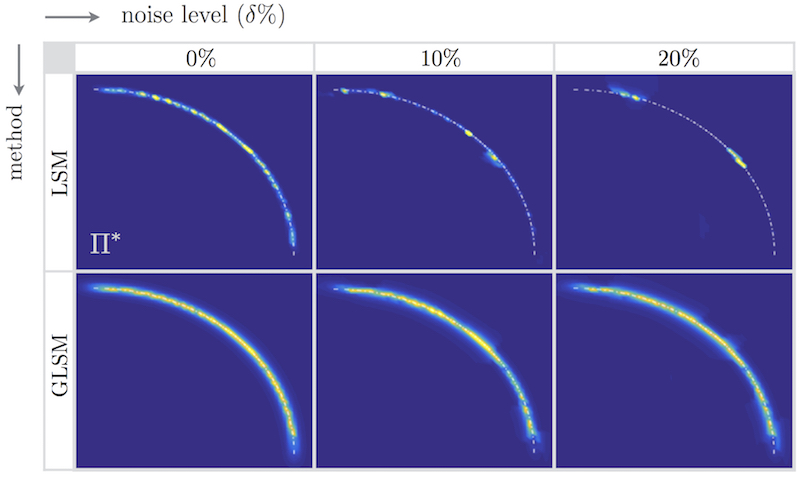} \vspace*{-2mm} 
\caption{Sensitivity to measurement noise for $\lambda_s/\ell = 0.3$: Full-aperture reconstruction of a cylindrical fracture, mid-section~$\Pi$, by the LSM indicator (top panels) and its GLSM counterpart (bottom panels).} \lb{F4}
\end{figure} 

\emph{Effect of the sensing aperture}. The ramifications of an incomplete aperture on the quality of fracture reconstruction are illustrated in Figs.~\ref{H1} and~\ref{H4}, where only \textcolor{black}{\emph{the ``upper'' half} of $\Omega$ in Fig.~\ref{SetNum}} is available for the purposes of illumination and observation.  More specifically, Figs.~\ref{H1} and~\ref{H4} depict the GLSM and LSM fields in the mid-section of $\Gamma$ at ``long'' ($\lambda_s/\ell = 1.3$) and ``short'' ($\lambda_s/\ell = 0.3$) excitation wavelengths, respectively, constructed from the half-aperture sensory data. While the loss of resolution in both GLSM and LSM maps is clear relative to Figs.~\ref{F1} and~\ref{F4}, it is noted that (for the problem under consideration) the GLSM indicator offers far better robustness to noise, providing acceptable reconstruction of~$\Gamma$ for~$\delta$ as high as~$0.1\|\textrm{\bf{F}}\|$.      
\begin{figure}[h!]
\center\includegraphics[width=0.66\linewidth]{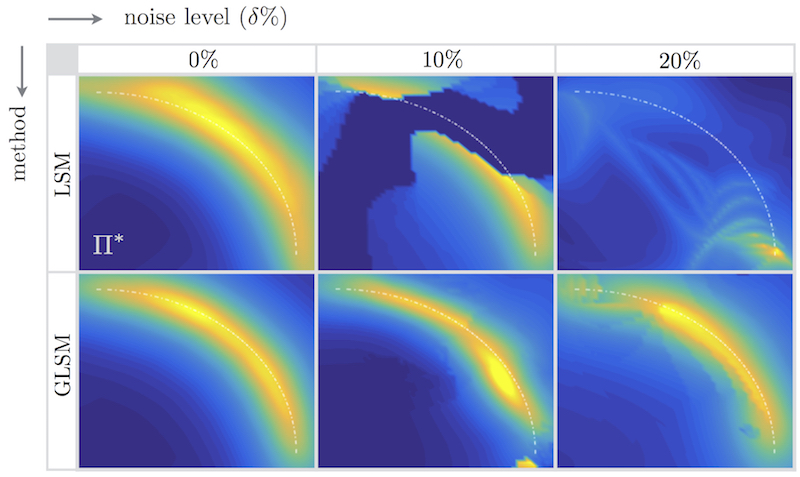} \vspace*{-2mm} 
\caption{Half-aperture reconstruction of a cylindrical fracture, mid-section $\Pi$, for $\lambda_s/\ell = 1.3$:~sensitivity of the LSM indicator (top panels) and its GLSM counterpart (bottom panels) to noise in the measurements.} \lb{H1} \vspace*{0mm} 
\center\includegraphics[width=0.66\linewidth]{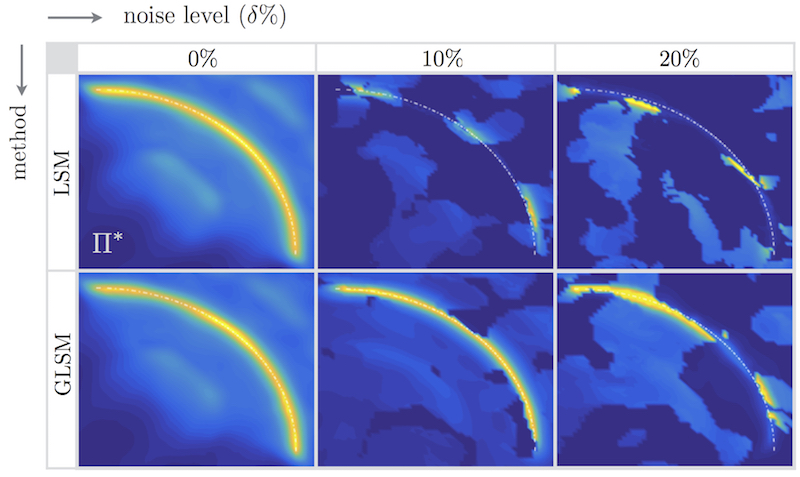} \vspace*{-2mm} 
\caption{Half-aperture reconstruction of a cylindrical fracture, mid-section $\Pi$, for $\lambda_s/\ell = 0.3$:~sensitivity of the LSM indicator (top panels) and its GLSM counterpart (bottom panels) to noise in the measurements.} \lb{H4} \vspace*{3mm} 
\center\includegraphics[width=0.7\linewidth]{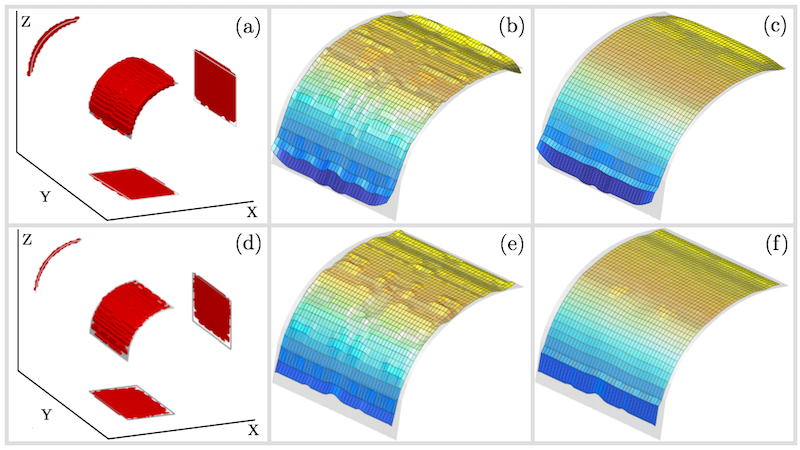} \vspace*{-2mm} 
\caption{Full-aperture 3D GLSM reconstruction for $\{\lambda_s/\ell = 1.3,\delta\% = 0.1\}$ (top) and $\{\lambda_s/\ell = 0.7,\delta\% = 0.05\}$ (bottom):~ GLSM indicator~\eqref{GLSM-Dgs} thresholded at 10\% (left), fracture surface as reconstructed from the 3D cloud of points (middle), and fracture reconstruction after the application of a mean filter (right).} \lb{3D}  \vspace*{-4mm}
\end{figure} 

\emph{3D reconstruction}. For completeness, Fig.~\ref{3D} illustrates the full-aperture GLSM reconstruction of $\Gamma$ inside the sampling region~$[-1,1]^3$, assuming $\lambda_s/\ell = 1.3$ and $\delta\% = 10$ (top panels) and $\lambda_s/\ell = 0.7$ and $\delta\% = 5$ (bottom panels). For clarity, the indicator maps are thresholded by $10\%$, i.e.~only the sampling points whose~$I^{\mathcal{G}_\sharp}(\bz)$ values are higher than ten percent of the global maximum value are shown (left panels). Then, a scattered interpolant is constructed based on thus obtained 3D cloud of points, giving an optimal reconstruction of the fracture surface. The latter is generated by (i) projecting the thresholded GLSM map onto a reference plane (the $X-Y$ plane in this example), and (ii) defining a suitable grid of points covering the projected area. This forms the sought-for input for the scattered interpolant providing a 3D reconstruction of the fracture interface, as shown in the middle panels of Fig.~\ref{3D}. Due in part to a scattered nature of the interpolant, thus obtained fracture surface will suffer from some artificial roughness -- that depends for example on the density of sampling points and an ad-hoc thresholding parameter. This issue may be mitigated by implementing a suitable spatial (e.g. moving average) filter, as shown in the right panels of Fig.~\ref{3D}.

\section{Conclusions} \label{Conc}

The Generalized Linear Sampling Method (GLSM) combined with the $F_\sharp$-factorization technique form a fast, yet robust, platform for the geometric reconstruction of heterogeneous (and dissipative) discontinuity surfaces from scattered wavefield data. It is illustrated that the GLSM indicator possesses little sensitivity to (the reasonable levels of) measurement noise -- that is comparable to the robustness of TS, while inheriting the top-tier localization property of the classical LSM, which guarantees a high-quality geometric characterization of the fracture -- notwithstanding the frequency regime of excitation and the unknown (generally heterogeneous) interfacial stiffness $\bK$. Such attributes carries a remarkable potential for developing a GLSM-based hybrid approach for not only geometric reconstruction of hidden fractures, but also identification of their interfacial condition (e.g.~retrieval of $\bK$ in the present work) from scattered field data. Furthermore, this approach may be naturally and rigorously extended to other sensing configurations and to more sophisticated background-domain geometries. It should also be noted that the analysis in this study does not require the fracture surface to be \emph{connected}, so one should be able to use the GLSM for simultaneous imaging of multiple fractures in the medium.  

\appendix

\section{Proof of equation~(\ref{kup1})}  \label{lem1-p}

As shown in~\cite{Kuprad1979}, the irrotational ($\btu^{\mathrm{p}}$) and solenoidal ($\btu^{\mathrm{s}}$) parts of a radiating wavefield $\btu=\btu_{\bphi}$ solving~\eqref{TR} exhibit the following asymptotic behavior as $r:=|\bxi|\to\infty$:
\beq\lb{A1}
\begin{aligned}
&\frac{\partial \btu^{\mathrm{p}}}{\partial r} - \text{i} k_p \btu^{\mathrm{p}} = O(r^{-2}),  \qquad \btu^{\mathrm{p}} = O(r^{-1}),                                             \\*[0.5 mm]
&  \frac{\partial\btu^{\mathrm{s}}}{\partial r} - \text{i} k_s \btu^{\mathrm{s}} = O(r^{-2}), \qquad \btu^{\mathrm{s}} = O(r^{-1}),                \\*[0.5 mm]
&\hat\bxi \sip \bC \colon \! \nabla \btu^{\mathrm{p}} - \text{i} k_p (\lambda+2\mu) \btu^{\mathrm{p}} = O(r^{-2}), \qquad 
 \hat\bxi \sip \bC \colon \! \nabla \btu^{\mathrm{s}} - \text{i} k_s \exs \mu  \btu^{\mathrm{s}} = O(r^{-2}), \qquad 
 \btu^{\mathrm{p}} \cdot \overline{ \btu^{\mathrm{s}}} = O(r^{-3})
\end{aligned}
\eeq  
For brevity, an auxiliary decomposition~$\bt = \hat\bxi \cdot \bC \colon\!\! \nabla \btu = \bt^{\mathrm{p}}+\bt^{\mathrm{s}}$ is adopted in the sequel, where $\bt^{\mathrm{p}} = \hat\bxi\cdot\bC \colon\!\!\nabla \btu^{\mathrm{p}}$ and $\bt^{\mathrm{s}} = \hat\bxi\sip\bC\colon \! \nabla \btu^{\mathrm{s}}$. In this setting, one finds that  
\beq\lb{Im(I)}
\begin{aligned}
\Im \lim_{r\to\infty} \int_{\partial{B_{r}}}   \overline\btu \cdot \bt \,\,   \textrm{d}S_{\bxi} &=&  
&\frac{1}{2\textrm{i}} \lim_{r\to\infty} \int_{\partial{B_{r}}}   
\Big\{ (\overline{\btu^{\mathrm{p}}}+\overline{\btu^{\mathrm{s}}}) \cdot (\bt^{\mathrm{p}}+\bt^{\mathrm{s}}) ~-~ 
(\btu^{\mathrm{p}}+\btu^{\mathrm{s}}) \cdot (\overline{\bt^{\mathrm{p}}}+\overline{\bt^{\mathrm{s}}})  \Big\} \,  \textrm{d}S_{\bxi}& \\
&=& & \lim_{r\to\infty} \int_{\partial{B_{r}}}  \Big\{k_p(\lambda+2\mu)\big(|\btu^{\mathrm{p}}|^2+
\Re(\btu^{\mathrm{p}}\sip\overline{\btu^{\mathrm{s}}})\big) ~+~   
k_s \mu \big(|\btu^{\mathrm{s}}|^2+\Re(\btu^{\mathrm{p}}\sip\overline{\btu^{\mathrm{s}}})\big) \Big\} \,\,  \textrm{d}S_{\bxi},& \\
&=& & \lim_{r\to\infty} \int_{\partial{B_{r}}}  \Big\{k_p(\lambda+2\mu) |\btu^{\mathrm{p}}|^2 ~+~   
k_s \mu |\btu^{\mathrm{s}}|^2\Big\} \,\,  \textrm{d}S_{\bxi},&
\end{aligned}
\eeq
which completes the proof.

\section{Elastodynamic fundamental stress tensor}  \label{stress-fund}

\textcolor{black}{
In dyadic notation, the elastodynamic fundamental stress tensor~\cite[e.g.][]{Ach2003} can be written as 
\beq\label{sfund2}
\bSig(\bx,\by) ~=~ \Sigma_{ij}^\ell(\bx,\by)\, \be_i\otimes\be_j\otimes\be_\ell, \qquad \bx,\by\in\mathbb{R}^3, \quad \bx\ne\by
\eeq
signifying the Cauchy stress tensor at $\bx$ due to point force $\be_\ell$ acting at $\by$, where 
\beq\label{sfund2}
\Sigma_{ij}^\ell(\bx,\by) = \frac{\lambda}{\lambda\!+\!2\mu}[G(k_pr)]_{,\ell} \, \delta_{ij} \,-\, 
\frac{2}{k_s^2}[G(k_pr)-G(k_sr)]_{,ij\ell} \,+\, [G(k_sr)]_{,i} \, \delta_{j\ell} \,+\, [G(k_sr)]_{,j} \, \delta_{i\ell}
\eeq 
and 
\[
r = |\bx-\by|, \qquad G(kr) := \frac{e^{\text{i}kr}}{4\pi r}, \qquad [f]_{,i} := \frac{\partial f}{\partial x_i}. 
\]
As shown in~\cite{Ach2003}, the far-field approximation of~\eqref{sfund2} as $|\bx|\to\infty$ reads
\beq\label{sfund3}
\Sigma_{ij}^{\ell,\infty}(\bx,\by) ~=~ \Sigma_{ij}^{\ell,p}(\bx,\by) \:+\: \Sigma_{ij}^{\ell,s}(\bx,\by), 
\eeq
where
\beq\lb{sfund4}
\begin{aligned}
& \Sigma_{ij}^{\ell,p}(\bx,\by) ~=~ \textrm{i} k_p\exs A_{ij\ell}^p \frac{e^{\textrm{i}k_p|\bx|}}{4\pi|\bx|}\, e^{-\textrm{i}k_p\hat{\bx}\cdot\by}, 
\quad & A_{ij\ell}^p& \,=\, \Big[\frac{2\mu}{\lambda\!+\!2\mu} \hat{x}_i \hat{x}_j \,+\,  \frac{\lambda}{\lambda\!+\!2\mu} \delta_{ij}\Big]\, \hat{x}_\ell, & \\
& \Sigma_{ij}^{\ell,s}(\bx,\by) ~=~ \textrm{i} k_s\exs A_{ij\ell}^s \frac{e^{\textrm{i}k_s|\bx|}}{4\pi|\bx|}\, e^{-\textrm{i}k_s\hat{\bx}\cdot\by}, 
\quad & A_{ij\ell}^s& \,=\, \delta_{i\ell} \hat{x}_j + \delta_{j\ell} \hat{x}_i - 2 \hat{x}_i \hat{x}_j \hat{x}_\ell. \hfill&
\end{aligned}
\eeq}

\section{Proof of Lemma~\ref{recip}}\label{Recip}

\textcolor{black}{Consider the orthonormal bases $(\be_1\!:=\!\hat\bxi,\be_2,\be_3)$ and $(\bh_1\!:=\!\bd,\bh_2,\bh_3)$, where $\hat\bxi$ and~$\bd$ denote respectively the directions of observation and plane-wave incidence. On representing the far-field pattern $\bv^\infty_\bq=\bv^\infty_{\bq_p}\oplus \bv^\infty_{\bq_s}$ (resp.~the polarization vector~$\bq=\bq_p\oplus\bq_s$) in the $\be-$ (resp.~$\bh-$) basis, definition~\eqref{w-inf} of the far-field kernel~$\bW^\infty(\bd,\hat\bxi)$ can be written in matrix form as
\beq\lb{mat1}
\bv^\infty_\bq(\bd,\hat\bxi) =
\left[\begin{array}{c}
\bv^\infty_{\bq_p}\sip\be_1 \\
0 \\
0 \\
0 \\
\bv^\infty_{\bq_s}\sip\be_2 \\
\bv^\infty_{\bq_s}\sip\be_3 
\end{array}\right] \;=\; 
\left[\begin{array}{cccccc}
W^{\infty}_{11}(\bd,\hat\bxi) & 0 & 0 & 0 & W^{\infty}_{12}(\bd,\hat\bxi) & W^{\infty}_{13}(\bd,\hat\bxi) \\
0 & 0 & 0 & 0 & 0 & 0 \\
0 & 0 & 0 & 0 & 0 & 0 \\
0 & 0 & 0 & 0 & 0 & 0 \\
0 & 0 & 0 & 0 & 0 & 0 \\
W^{\infty}_{21}(\bd,\hat\bxi) & 0 & 0 & 0 & W^{\infty}_{22}(\bd,\hat\bxi) & W^{\infty}_{23}(\bd,\hat\bxi) \\
W^{\infty}_{31}(\bd,\hat\bxi) & 0 & 0 & 0 & W^{\infty}_{32}(\bd,\hat\bxi) & W^{\infty}_{33}(\bd,\hat\bxi) 
\end{array}\right]
\left[\begin{array}{c}
\bq_p\sip\bh_1 \\
0 \\
0 \\
0 \\
\bq_s\sip\bh_2 \\
\bq_s\sip\bh_3 
\end{array}\right]. 
\eeq
In this setting, the reciprocity statement~\eqref{W-recip} can be rewritten as 
\beq\lb{W-recip2}
W^{\infty}_{ij}(\bd,\hat\bxi) ~=~ W^{\infty}_{ji}(-\hat\bxi,-\bd), \qquad i,j=1,2,3.
\eeq
This section aims to extend Lemma~1 in~\cite{Dassios1987} to cater for the scattering problem~\eqref{GE}-\eqref{KS} with its particular boundary condition, \eqref{contact}, at the fracture interface~$\Gamma$. With such result in place, the distilled reciprocity claim~\eqref{W-recip2} follows immediately as a consequence of Theorem~1 and its corollaries in~\cite{Dassios1987}.} To this end, consider two distinct \emph{total fields} 
\[
\bfpsi_1 = \bu\ff_1+\bv_1, \qquad \bfpsi_2 = \bu\ff_2+\bv_2,
\]
where $\bv_j\in H^1_{\mathrm{loc}}(\R^3 \backslash \Gamma)^3$ satisfies \eqref{GE}-\eqref{KS} with $\bu\ff=\bu\ff_j$ $(j\!=\!1,2)$. On adopting Twersky's notation~\cite{Twersky1962}  
\[
\lbrace \bfpsi_1, \bfpsi_2 \rbrace_{S}  ~:=~ \int_{S} \big\lbrace \bfpsi_1 \cdot \exs \bt( \bfpsi_2) ~-~ \bfpsi_2 \cdot \exs \bt( \bfpsi_1) \big\rbrace  \,\, \text{d}S_{\bxi}, 
\]
Lemma~1 in~\cite{Dassios1987} states that $\lbrace \bfpsi_1, \bfpsi_2 \rbrace_{\partial B_R}=0$, where $B_R$ is a ball of radius~$R$ sufficiently large so that $\Gamma\!\subset\!B_R$. By substituting the Navier equation 
\[
\nabla(\bC\!:\!\nabla\bfpsi_j) \,+\, \rho\,\omega^2\bfpsi_j \,=\, \boldsymbol{0} \quad\text{in}\quad\mathbb{R}^3\backslash\Gamma, \qquad j=1,2
\]
into Betti's third formula~\cite{Kuprad1979} written for domain $B_R\backslash\Gamma$, one finds that  
\beq\lb{Twersky}
\lbrace \bfpsi_1, \bfpsi_2 \rbrace_{\partial\Bo \cup \exs \Gamma} ~=~ \int_{\Bo \backslash \Gamma} \Big\lbrace \bfpsi_1 \cdot \big[ \nabla \sip (\bC \colon \! \nabla  \bfpsi_2) \big] ~-~ \bfpsi_2 \cdot \big[ \nabla \sip (\bC \colon \! \nabla  \bfpsi_1) \big] \Big\rbrace  \,\, \text{d}S_{\bxi} ~=~ 0
\eeq
where, thanks to the jump condition on~$\Gamma$ in~\eqref{GE} and contact law~\eqref{contact}, one has 
\beq\lb{Gama}
\begin{aligned}
\lbrace \bfpsi_1, \bfpsi_2 \rbrace_{\Gamma} ~=\, & \int_{\Gamma} \big\lbrace \llbracket  \bfpsi_1 \rrbracket \cdot  \exs \bt( \bfpsi_2) - \llbracket  \bfpsi_2 \rrbracket \cdot \exs \bt( \bfpsi_1) \big\rbrace \, \text{d}S_{\bxi}  \\*[1 mm]
~=\, & \int_{\Gamma} \Big\lbrace \llbracket  \bfpsi_1 \rrbracket \sip \bK \sip  \llbracket  \bfpsi_2 \rrbracket - \llbracket  \bfpsi_2\rrbracket \sip \bK \sip  \llbracket  \bfpsi_1 \rrbracket \Big\rbrace \, \text{d}S_{\bxi}. 
\end{aligned}
\eeq
Due to symmetry of~$\bK$, one has~$\lbrace \bfpsi_1, \bfpsi_2 \rbrace_{\Gamma} = 0$ and consequently~$\lbrace \bfpsi_1, \bfpsi_2 \rbrace_{\partial B_R}=0$ thanks to~\eqref{Twersky}.  \hfill$\blacksquare$

\section{Proof of Lemma~\ref{H*}}  \label{H*pruf}

\textcolor{black}{With reference to the Herglotz operator $\mathcal{H} \colon L^2(\OOd)^3 \rightarrow H^{-1/2}(\Gamma)^3$ given by~\eqref{oH} and a fracture opening displacement (FOD) profile~$\ba\in\tilde{H}^{1/2}(\Gamma)$, consider the duality product 
\beq\lb{adj}
\big\langle  \mathcal{H}(\bg), \ba \big\rangle ~=~ 
\int_\Gamma \,\, \bar\ba \cdot \bt(\btu_{\bg}) \,\, \text{d}S_{\by}. 
\eeq
Thanks to~\eqref{HW} and the linearity of~$\bt$,} the right-hand side of~(\ref{adj}) can be recast as
\[
\begin{aligned}
\int_\Gamma \,\,  \bar\ba \cdot \bt(\btu_{\bg})  \,\, \text{d}S_{\by} ~=~ & \int_{\OOd} \int_\Gamma \,\,  \bar\ba(\by) \cdot \bt \exs \big( \exs  \bg(\bd)  \cdot (\bd  \otimes  \bd \exs)  \,\, e^{\text{i}k_p \bd \cdot \by} \exs \big) \,\, \text{d}S_{\by} \,\, \text{d}S_{\bd}\\*[2 mm]
& ~+~ \int_{\OOd} \int_\Gamma \,\,  \bar\ba(\by) \cdot \bt \exs \big( \exs  \bg(\bd) \cdot (\bI - \bd  \otimes  \bd \exs)  \,\, e^{\text{i}k_s \bd \cdot \by} \exs \big) \,\, \text{d}S_{\by} \,\, \text{d}S_{\bd}.
\end{aligned}
\]
On recalling that for arbitrary smooth surface~$S$
\[
\bt(\bu) \;=\; \bn\sip\bC\!:\!\nabla\bu \;=\; \lambda \, \bn\exs \nabla\nxs\sip \bu \,+\, 2\mu \, \bn \cdot\! \nabla \bu \,+\,\mu \, \bn\!\times\! \nabla \!\times\! \bu ~\quad \text{on}\quad S
\]
where~$\bC$ is given by~\eqref{bC} and $\bn$ is the unit normal on~$S$, one finds that
\[
\begin{aligned}
& \bar\ba \cdot \bt \exs \big( \exs  \bg(\bd)  \cdot (\bd  \otimes  \bd \exs)  \,\, e^{\text{i}k_p \bd \cdot \by} \exs \big) ~=~  \bg(\bd) \cdot \overline{ (-\text{i}k_p) \exs e^{-\text{i}k_p \bd \cdot \by} \,\, \Big\lbrace  2\mu \exs (\ba\sip\bd) (\bn\sip\bd)  \,+\, \lambda \exs (\ba \sip \bn) \Big\rbrace \, \bd},  \\*[2 mm]
& \bar\ba \cdot \bt \exs \big( \exs  \bg(\bd) \cdot (\bI - \bd  \otimes  \bd \exs)  \,\, e^{\text{i}k_s \bd \cdot \by} \exs \big) ~=~ \bg(\bd) \cdot \overline{(-\textrm{i}k_s) \exs e^{-\textrm{i} k_s \bd \cdot \by} \,\, \bd \times \nxs \Big\lbrace  \mu \exs (\ba\sip\bd) (\bn \!\times \! \bd) \,+\, \mu \exs (\bn\sip\bd)(\ba \!\times \!\bd)   \Big\rbrace}.
\end{aligned}
\]
As a result, 
\[
\begin{aligned}
\int_\Gamma \,\,  \bar\ba \cdot \bt(\btu_{\bg})  \,\, \text{d}S_{\by} ~=~  & \int_{\OOd} \bg(\bd) \cdot \overline{(-\text{i}k_p) \, \bd  \int_\Gamma \, \Big\lbrace \lambda \exs (\ba \sip \bn) \,+\, 2\mu \exs (\bn \sip \bd) ( \ba \sip \bd)  \Big\rbrace  \, e^{-\text{i}k_p \bd \cdot \by} \,\, \text{d}S_{\by} } \,\, \text{d}S_{\bd}~+~ \\*[1 mm]
& \int_{\OOd} \bg(\bd) \cdot \overline{(-\text{i}k_s) \, \bd \times \nxs  \int_\Gamma \, \Big\lbrace   \mu \exs(\ba \!\times\!\bd)(\bn\sip\bd) \,+\, \mu \exs  (\bn \!\times\! \bd) (\ba \sip \bd)  \Big\rbrace \, e^{-\textrm{i} k_s \bd \cdot \by} \,\, \text{d}S_{\by}} \,\, \text{d}S_{\bd}. 
\end{aligned}
\]
\textcolor{black}{By virtue of~\eqref{herden} and~\eqref{far-field} which verify $\langle\exs\bg,\bv^\infty\rangle = \langle\exs\bg_p,\bv_p^\infty\rangle + \langle\exs\bg_s,\bv_s^\infty\rangle$, one finds that 
\[
\big\langle\mathcal{H}(\bg), \,\ba\big\rangle ~=~\, 
\big\langle\bg, \,\mathcal{H}^*(\ba)\big\rangle
\]}
which establishes~(\ref{Hstar}). \hfill$\blacksquare$ 

\section{Proofs of Theorem~\ref{GLSM1} and Theorem~\ref{GLSM2}}  \label{GLSM*pruf}

\begin{proof57} Consider the following:
\begin{itemize}
\item~Let~$\bPhi_L^\infty \in Range(\mathcal{H}^*)$. By definition, $\exists \exs \bfpsi \in \overline{Range(T)}$ such that $\mathcal{H}^* \bfpsi = \bPhi_L^\infty$. Then, by recalling the continuity of $T^{-1}\nxs$ \textcolor{black}{(Lemma~\ref{T-invs0})} and the range denseness of $\mathcal{H}$ \textcolor{black}{ (Lemma~\ref{H*p})}, one may find $\bg\Zsub \in L^2(\OOd)^3$ for every $\alpha>0$ such that $\norms{\nxs\mathcal{H}\bg\Zsub - T^{-1} \bfpsi \nxs}^2  \, < \alpha$. Now, let us observe that  
\begin{description}
\item{i)}~by the continuity of $\mathcal{G}=\mathcal{H}^* T$ \textcolor{black}{(Lemma~\ref{comp_G})}, one has 
\[
\norms{ \nxs F \bg\Zsub - \bPhi_L^\infty \nxs}^2 \!~\leqslant~ \alpha \! \norms{\nxs\mathcal{G}\nxs}^2; 
\]
\item{ii)}~\textcolor{black}{the boundedness i.e. continuity of $\mathfrak{T}$ (see Lemma~\ref{I{T}>0} and~\eqref{Tsdef})} implies that
\[
| (  \bg\Zsub , B \bg\Zsub) |  ~\leqslant \! ~ \norms{\nxs \mathfrak{T} \nxs} \norms{\! \mathcal{H} \bg\Zsub \!}^2 ~\! < ~ 2 \nxs \norms{\nxs \mathfrak{T} \nxs} \nxs (\alpha \,\,+ \norms{\!T^{-1} \bfpsi\!}^2); 
\]
\item{iii)}~thanks to the definitions of~$\mathfrak{j}_\alpha(\bPhi_L^\infty)$ and $\bg_\alpha^L$, one has
\[
\mathfrak{J}_{\alpha}(\bPhi_L^\infty;\,\bg_\alpha^L) \,-\, \mu \alpha \,\,\,\leqslant \,\,\, \mathfrak{j}_\alpha(\bPhi_L^\infty) \,\,\leqslant ~ \! 
\norms{\nxs F\bg\Zsub-\bPhi_L^\infty \nxs}^2 \,+\,\,\exs \alpha \exs | ( \bg\Zsub, B \bg\Zsub )|.
\] 
\end{description} 
As a result, it immediately follows that 
\beq\lb{inq1}
\alpha \exs | (  \bg_\alpha^L, B \bg_\alpha^L ) | \,\,\,\leqslant \,\,\, \mathfrak{J}_{\alpha}(\bPhi_L^\infty;\bg_\alpha^L) \,\,\leqslant \,\, \mu \alpha  \,+\, \alpha \! \norms{\nxs\mathcal{G}\nxs}^2 \,+\,\,\exs 2 \alpha \! \norms{\nxs \mathfrak{T} \nxs} \! (\alpha \,\,+ \norms{\!T^{-1} \bfpsi\!}^2),
\eeq
whereby~$\limsup\limits_{\alpha \rightarrow 0} |( \bg_\alpha^L, B \bg_\alpha^L ) | < \infty$ which implies $\liminf\limits_{\alpha \rightarrow 0} |(  \bg_\alpha^L, B \bg_\alpha^L ) | < \infty$. 

\item~Next, let~$\bPhi_L^\infty\!\not\in\!Range(\mathcal{H}^*)$.~Let us by contradiction assume that $\liminf\limits_{\alpha \rightarrow 0} |(  \bg_\alpha^L, B \bg_\alpha^L ) | < \infty$; then, for some constant $c\!>\!0$ independent of $\alpha$, one has~$|( \bg_\alpha^L, B \bg_\alpha^L )|\!<\!c$ for an extracted subsequence of $\bg_\alpha^L$. The coercivity of $\mathfrak{T}$ then implies that $\mathcal{H}\bg_\alpha^L$ is also bounded. As $H^{-1/2}(\Gamma)^3$ is reflexive, one may suppose that up to an extracted subsequence, $\mathcal{H}\bg_\alpha^L$ weakly converges to some $T^{-1} \bfpsi \in H^{-1/2}(\Gamma)^3$. In fact, $T^{-1} \bfpsi  \in \overline{Range(\mathcal{H})}$ since the latter set is convex. Now, since $\mathcal{G}$ is compact, $\mathcal{G}\mathcal{H}\bg_\alpha^L$ strongly converges to $\mathcal{G}T^{-1} \bfpsi = \mathcal{H}^*\bfpsi $ as $\alpha\!\rightarrow\!0$. Recalling the definition of $\mathfrak{J}_{\alpha}(\bPhi_L^\infty;\bg_\alpha^L)$ and the fact that $\mathfrak{j}_\alpha(\bPhi_L^\infty) \rightarrow 0$ as $\alpha \rightarrow 0$ thanks to the range denseness of~$F$, one may observe that $\norms{\nxs F\bg_\alpha^L-\bPhi_L^\infty \nxs}^2 \,\leqslant \, \mathfrak{J}_{\alpha}(\bPhi_L^\infty;\bg_\alpha^L) \,\leqslant \, \mathfrak{j}_\alpha(\bPhi_L^\infty) + \mu \alpha \rightarrow 0$ as $\alpha \rightarrow 0$. Thus, $\mathcal{H}^*\bfpsi  = \bPhi_L^\infty$ which is a contradiction. Accordingly, $\bPhi_L^\infty \not\in Range(\mathcal{H}^*)$ necessitates $\,\liminf\limits_{\alpha \rightarrow 0} |(\bg_\alpha^L, B \bg_\alpha^L)| \!=\! \infty$ which in turn implies $\limsup\limits_{\alpha \rightarrow 0} |(\bg_\alpha^L, B \bg_\alpha^L)| \!=\! \infty$.     
\end{itemize}
\end{proof57} 

\begin{proof59}
The logic of this proof follows that of Theorem~\ref{GLSM1}, and entails the following steps. 
\begin{itemize}
\item~Let $\bPhi_L^\infty \in Range(\mathcal{H}^*)$ so that~$\mathcal{H}^* \bfpsi = \bPhi_L^\infty$ for some~$\bfpsi \in \overline{Range(T)}$. Define for every $\alpha\!>\!0$ independent of $\delta$, density~$\bg\Zsub \in L^2(\OOd)^3$ such that $\|\mathcal{H}\bg\Zsub - T^{-1} \bfpsi \nxs\|^2 < \alpha$, and set $\delta>0$ sufficiently small so that
\[
\textcolor{black}{\delta \big\{2 \alpha\|\bg\Zsub \|^2 \,+\, \delta\|\bg\Zsub \|^2  \,+\, 2 \|F\bg\Zsub -\bPhi_L^\infty\|\exs \|\bg\Zsub \|\big\} \!\!~\leqslant\, \alpha.} 
\]
With reference to~\eqref{JadJa}, one finds
\beq\lb{JJJa}
\mathfrak{J}_\alpha^\delta(\bPhi_L^\infty;\,\bg_{\alpha,\delta}^L) \,\,\,\leqslant\,\,\, \mathfrak{J}_\alpha^\delta(\bPhi_L^\infty;\,\bg\Zsub) \,\,\,\leqslant\,\,\, \mathfrak{J}_\alpha(\bPhi_L^\infty;\,\bg\Zsub) \,+\, \alpha.
\eeq
On recalling the bound in~\eqref{inq1} on~$ \mathfrak{J}_\alpha$, this yields 
\[
\alpha \exs \big(\exs|(\bg_{\alpha,\delta}^L, B^\delta \bg_{\alpha,\delta}^L)| \,+\, \delta \! \norms{\nxs \bg_{\alpha,\delta}^L \nxs}^2 \! \big) \,\,\,\leqslant\,\,\, \mathfrak{J}_\alpha^\delta(\bPhi_L^\infty;\bg_{\alpha,\delta}^L) \,\,\,\leqslant\,\,\, \textcolor{black}{(\mu\!+\!1)\alpha}  \,+\, \alpha \! \norms{\nxs\mathcal{G}\nxs}^2\!\! \,\,+\,\, 2 \alpha \! \norms{\nxs \mathfrak{T} \nxs} \! (\alpha \,\,+ \norms{\!T^{-1} \bfpsi\!}^2),
\]
which guarantees that~$\limsup\limits_{\alpha \rightarrow 0}\limsup\limits_{\delta \rightarrow 0}\big(\exs |( \bg_{\alpha,\delta}^L, B^\delta \bg_{\alpha,\delta}^L)| \exs+\exs \delta \! \norms{\! \bg_{\alpha,\delta}^L \!}^2 \nxs \big) < \infty$.

\item~Let $\bPhi_L^\infty\!\not\in Range(\mathcal{H}^*)$, and assume to the contrary that $\liminf\limits_{\alpha \rightarrow 0}\liminf\limits_{\delta \rightarrow 0}\big(|(\bg_{\alpha,\delta}^L, B^\delta \bg_{\alpha,\delta}^L)|+\delta\|\bg_{\alpha,\delta}^L\|^2\nxs \big)<\infty$. Using the coercivity of $\mathfrak{T}$ and triangle inequality, one finds     
\[
c \norms{\mathcal{H}\bg_{\alpha,\delta}^L}^2 \,\,\,\leqslant\,\, | (  \bg_{\alpha,\delta}^L, B \bg_{\alpha,\delta}^L ) | \,\,\,\leqslant\,\, | (  \bg_{\alpha,\delta}^L, B^\delta \bg_{\alpha,\delta}^L ) | \,+\, \delta \! \norms{\nxs \bg_{\alpha,\delta}^L \nxs}^2,
\]
whereby $\limsup\limits_{\alpha \rightarrow 0}\limsup\limits_{\delta \rightarrow 0} \norms{ \mathcal{H}\bg_{\alpha,\delta}^L}^2  < \infty$. Then, there exists a subsequence $(\alpha',\delta(\alpha')\!<\!\alpha')$ such that $\alpha' \rightarrow 0$ and $\norms{\!\mathcal{H}\bg_{{\alpha'\nxs},{\delta(\alpha')}}^L\!\!}^2$ is bounded independently from $\alpha'$. In light of Lemma~\ref{min-Jad}, one may design this subsequence such that $\mathfrak{J}_{\alpha'}^{\delta(\alpha')}(\bPhi_L^\infty;\bg_{{\alpha'\nxs},{\delta(\alpha')}}^L)\rightarrow 0$ as $\alpha' \rightarrow 0$, and thus~$\norms{F^\delta \bg_{{\alpha'\nxs},{\delta(\alpha')}}^L - \, \bPhi_L^\infty} \rightarrow 0$ as $\alpha' \rightarrow 0$. The compactness of $\mathcal{H}^*$ \textcolor{black}{and boundedness of~$\mathfrak{T}$} imply that a subsequence of $\mathcal{H}^* \mathfrak{T}  \mathcal{H}  \exs \bg_{{\alpha'\nxs},{\delta(\alpha')}}^L\!$ converges to some $\mathcal{H}^* \bfpsi$ in $ L^2(\OOd)^3$. The uniqueness of this limit implies that $\mathcal{H}^* \bfpsi = \bPhi_L^\infty$, which is a contradiction. 
\end{itemize} 
\end{proof59} 

\bibliography{Paper_GLSM}

\end{document}

%% file: macros.tex
\newlength{\kaka}

\newcommand{\ahref}[2]{}

\newcommand{\Zsup}{^{\text{\tiny o}}}
\newcommand{\Zsub}{_{\text{\tiny o}}}
\newcommand{\zsub}{_{\!\text{\tiny o}}}

\newcommand{\ff}{^{\text{\tiny f}}}

\newcommand{\beq}{\begin{equation}}
\newcommand{\eeq}{\end{equation}}
\newcommand{\lb}{\label}

\newcommand{\bea}{\begin{eqnarray}}
\newcommand{\eea}{\end{eqnarray}}
\newcommand{\bxr}{\begin{array}}
\newcommand{\exr}{\end{array}}

\newcommand\exs{\hspace*{0.4mm}}

\newcommand\nxs{\hspace*{-0.2mm}}

\newenvironment{proof57}{\paragraph*{Proof of Theorem 5.7.}}{\hfill$\blacksquare$}
\newenvironment{proof59}{\paragraph*{Proof of Theorem 5.9.}}{\hfill$\blacksquare$}

\newcommand{\norms}[1]{\parallel\! #1 \!\parallel}

\newcommand{\bSig} {\boldsymbol{\Sigma}}
\newcommand{\bW} {\boldsymbol{W}}
\newcommand{\bC} {\boldsymbol{C}}
\newcommand{\bK} {\boldsymbol{K}}
\newcommand{\bV} {\boldsymbol{V}}
\newcommand{\bR} {\boldsymbol{\sf R}}

\newcommand{\bn} {\boldsymbol{n}}
\newcommand{\ba} {\boldsymbol{a}}
\newcommand{\bb} {\boldsymbol{b}}

\newcommand{\bx} {\boldsymbol{x}}
\newcommand{\by} {\boldsymbol{y}}
\newcommand{\be} {\boldsymbol{e}}
\newcommand{\bh} {\boldsymbol{h}}
\newcommand{\bz} {\boldsymbol{z}}
\newcommand{\bg} {{\boldsymbol{g}}}
\newcommand{\bq} {{\boldsymbol{q}}}

\newcommand{\bfT} {\boldsymbol{T}}

\newcommand{\bd} {\boldsymbol{d}}
\newcommand{\bI} {\boldsymbol{I}}

\newcommand{\pff}{\boldsymbol{u}^{\textit{i}}}

\newcommand{\bfpsi}{\boldsymbol{\psi}}
\newcommand{\bfPsi}{\boldsymbol{\Psi}}

\newcommand{\Tcal}{{\sf T}}

\newcommand{\sip} {\!\cdot\!}

\newcommand{\OO}{\Omega}

\newcommand{\OOd}{\Omega_{\bd}}

\newcommand{\herg}{\mathcal{H}}

\newcommand{\Bo}{\mathcal{B}_2}

\newcommand{\bzero}{\boldsymbol{0}}

\newcommand{\bu} {\boldsymbol{u}}

\newcommand{\bt} {{\boldsymbol{t}}}

\newcommand{\bv} {\boldsymbol{v}}

\newcommand{\bxi} {\boldsymbol{\xi}}
\newcommand{\bxio} {\bxi\Zsup}

\newcommand{\dbv} {\llbracket {\bv}\rrbracket}

\newcommand{\btu} {\text{\bf{u}}}

\newcommand{\bw} {\boldsymbol{w}}

\newcommand{\bPhi}{\boldsymbol{\Phi}}

\renewcommand{\Bo}{{B_R}}
\newcommand{\dualBR}[2]{\left< #1, #2 \right>_{\partial B_R}}
\newcommand{\dualGA}[2]{\left< #1, #2 \right>_{\Gamma}}
\newcommand{\bphi}{{\boldsymbol{\varphi}}}
\newcommand{\bpsi}{{\boldsymbol{\psi}}}